\documentclass[a4paper,12pt]{article}
\usepackage{hyperref}
\usepackage[english]{babel}
\usepackage[T1]{fontenc}
\usepackage[utf8]{inputenc}
\usepackage{amsfonts}
\usepackage{amssymb}
\usepackage{amsmath}
\usepackage{graphicx}
\usepackage{cite}
\usepackage{epsfig}
\usepackage{psfrag}
\usepackage{tikz}
\usetikzlibrary{matrix,arrows,decorations.pathmorphing,positioning} 
\usepackage{url}
\usepackage{amsthm}
\usepackage{enumerate}

\newtheorem{proposition}{Proposition}[section]
\newtheorem{definition}{Definition}[section] 
  
\newtheorem{corollary}{Corollary}[section] 
 
\newtheorem{remark}{Remark}[section]  
\newtheorem{example}{Example}[section]

\date{}

\newcommand{\ket}[1]{|#1\rangle}
\newcommand{\bra}[1]{\langle #1|}
\DeclareMathOperator{\variety}{Var}

\title{\bf Quantum Gates from Wolfram Model Multiway Rewriting Systems} 

\author{Furkan Semih Dündar$^{1, 2, 3, }$\footnote{\url{f.semih.dundar@gmail.com}}  {} {}   Xerxes D. Arsiwalla$^{3, }$\footnote{\url{x.d.arsiwalla@gmail.com}} {} {}  Hatem Elshatlawy$^{3,  }$\footnote{\url{hatemelshatlawy@gmail.com}}  \\  
{}  \\
{\it \small $^{1}$Department of Mechanical Engineering, Faculty of Engineering, Amasya University, Turkey}\\ 
{\it \small $^{2}$Department of Physics, Faculty of Science, Sakarya University, Turkey}\\
{\it \small $^{3}$Wolfram Institute for Computational Foundations of Science, Illinois, USA}    
}

\begin{document}

\maketitle

\begin{abstract}

We show how representations of finite-dimensional quantum operators can be constructed using nondeterministic rewriting systems. In particular, we investigate Wolfram model multiway rewriting systems based on string substitutions. Multiway systems were proposed by S.~Wolfram as generic model systems for multicomputational processes, emphasizing their significance as a foundation for modeling complexity, nondeterminism, and branching structures of measurement outcomes. Here, we investigate a specific class of multiway systems based on cyclic character strings with a neighborhood constraint - the latter called Leibnizian strings. We show that such strings exhibit a Fermi-Dirac distribution for expectation values of occupation numbers of character neighborhoods. A Leibnizian string serves as an abstraction of a $N$-fermion system. A multiway system of these strings encodes causal relations between rewriting events in a nondeterministic manner. The collection of character strings realizes a $\mathbb{Z}$-module with a symmetric $\mathbb{Z}$-bilinear form. For discrete spaces, this generalizes the notion of an inner product over a vector field. This admits a discrete analogue of the path integral and a $S$-matrix for multiway systems of Leibnizian strings. The elements of this $S$-matrix yield transition amplitudes between states of the multiway system based on an action defined over a sequence of Leibnizian strings. We then show that these $S$-matrices give explicit representations of quantum gates for qubits and qudits, and also circuits composed of such gates. We find that, as formal models of nondeterministic computation, rewriting systems of Leibnizian strings with causal structure encode representations of the CNOT, $\pi/8$, and Hadamard gates. Hence, using multiway systems one can represent quantum circuits for qubits.

\vspace{2em}  
\end{abstract}

\clearpage

\tableofcontents

\clearpage

\section{Introduction}

Symbolic systems broadly constitute a framework of abstract entities, or symbols, and a collection of rules for manipulating these symbols. That is, a syntax equipped with a generative grammar or an algebra governing the composition of symbols. Although these symbols and their operations are purely syntactic, in practical use cases, they often refer to (in terms of formally admissible semantics) specific knowledge or information about objects and relations in the world (real or virtual). Common examples of symbolic systems are formal language, logic, algorithms, just to name a few. In this work, we focus on computational systems, in particular, nondeterministic rewriting systems based on string substitutions. Although much is known about these systems as formal models of computation \cite{huet1980confluent, dershowitz1990rewrite, baader1998term, bezem2003term}, here we seek to explore how rewriting systems encode representations of finite-dimensional quantum operators. A symbolic representation of the latter is often useful for the design of algorithms and protocols in quantum information and computation. 

Prior approaches towards a symbolic representation of finite-dimensional quantum theory include category-theoretic process theories such as `categorical quantum mechanics' (CQM) based on string diagrams (where the ``strings'' of CQM refer to connectors between operators) \cite{abramsky2009categorical, coecke2018picturing}, and the related ZX-calculus \cite{coecke2011interacting, coecke2012tutorial} (one may also consider a multiway system of ZX diagrams, as shown in \cite{gorard2020zx, gorard2021zx, gorard2021fast}). The interesting feature about such symbolic approaches to quantum theory is that they offer a formulation that is independent of any background spacetime geometry. Instead, they make use of the more fundamental notion of causal order. This approach has led to various diagrammatic calculi that have proven useful for quantum computing protocols \cite{coecke2011interacting, coecke2018picturing, zx_qudit_Wang, zx_qudit_Poor}. This is very much the spirit of the computational paradigm we espouse here, which is  background-independent and makes use of the causal structure of the rewriting system. On the other hand, our formalism differs from some of the aforementioned ones in two important ways: (i) rather than working with instantaneous system configurations, here we consider a statistical mechanics approach based on an ensemble of paths referring to sequences of rewriting events; and (ii) the building blocks of our formalism are not themselves quantum states or operators (or direct references to those), but instead are syntactic entities of an abstract computational system. 

More specifically, we work with nondeterministic rewriting systems based on cyclic character strings. These strings are updated following all specified rewriting rules, and in all possible orders. Such a system is called a multiway rewriting system of character strings. The term ``multiway system'' was coined by S.~Wolfram in \cite{NKS}, and also refers to more generic rewriting systems based on graph and hypergraph rewriting. Rewriting systems are models of computation that are equivalent to Turing machines, and are commonly used in theorem proving protocols  \cite{dershowitz1990rewrite, baader1998term, bezem2003term}.  In \cite{wolfram2020} hypergraph rewriting systems were proposed as the computational basis underlying contemporary theories of physics. Subsequently, a mathematical formalization based on higher category theory established the construction of topological spaces from multiway rewriting systems \cite{arsiwalla2020homotopic, arsiwalla2021homotopies, arsiwalla2024pregeometric}. A review  of the implications and challenges of this computational approach for reformulating contemporary theories of physics can be found in   \cite{arsiwalla2023pregeometry, rickles2023ruliology, elshatlawy2025towards, rickles2025quantum}.

Initial investigations seeking quantum theory from rewriting systems have been discussed in \cite{wolfram2020, gorard2020some, gorard2020zx, gorard2021zx}. In this work, we go deeper into the specific issue of extracting finite-dimensional quantum operators from multiway systems. We solve this problem explicitly and compute representations of quantum gates from multiway data.  For this, we consider an ensemble of cyclic character strings. Furthermore, we restrict ourselves to the so-called \emph{Leibnizian} strings. These are inspired by a monadic interpretation of discernibility. We make use of a complexity measure on strings known as the Barbour-Smolin-Deutsch variety or the \emph{BSD variety}, which was first proposed in \cite{barbour1992extremal}, and more recently elaborated in \cite{dundar2024case}. The ensemble of Leibnizian strings turns out to exhibit a Fermi-Dirac distribution for expectation values of occupation numbers referring to string neighborhoods. By considering only the Leibnizian strings in the multiway system, we study paths on the multiway graph formed by sequences of state transitions. We consider a summation over  paths formed by sequences of such strings. This serves as a discrete analogue of the quantum mechanical path integral over trajectories in phase space. Using this, we then construct $S$-matrices whose entries consist of transition amplitudes weighted over an ensemble of paths over the multiway graph based on an action defined over a sequence of Leibnizian strings. These $S$-matrices provide us with a representation of unitary operators.

The $S$-matrices we construct here capture the underlying statistical mechanics involved in the update sequences of a nondeterministic rewriting system. The nondeterministic causal structure (which itself follows from the rewriting rules) plays a crucial role for the purposes of determining transition amplitudes between multiway states. This suggests a statistical perspective on nondeterministic computation. Even though this work investigates only the very specific aspect of representations of quantum operators from nondeterministic computational systems, it is worth noting the broader significance of computational approaches that probe the foundations of quantum theory. In particular, the work of G.~`t Hooft  \cite{thooft, t2010classical, hooft2023ontological} and C.~Wetterich   \cite{wetterich2022probabilistic, wetterich2024probabilistic, kreuzkamp2025quantum}, both of whom have proposed the emergence of quantum theory from cellular automata, is the subject of considerable interest. The former advocates fundamental determinism, while the latter favors a statistical mechanics approach based on probabilistic cellular automata, and posits that quantum theory emerges from an underlying classical statistical system.

The paper is organized as follows. Section~\ref{sec:prelim} provides preliminary 
definitions and concepts. In Section~\ref{sec:Leibnizian_strings_info_theory}, we 
define Leibnizian strings and develop an information-theoretic characterization of 
their structure. Section~\ref{sec:nd_rewrite} introduces the Wolfram model, abstract rewriting systems and multiway graphs. Section~\ref{sec:phys_max_var_paths} identifies physical and 
maximal variety paths. In Section~\ref{sec:fermionic}, we demonstrate that Leibnizian strings exhibit Fermi-Dirac statistics for occupation numbers of monadic views. 
Section~\ref{sec:s_matrix} develops the path-sum formalism and constructs $S$-matrices 
for multiway evolution. Section~\ref{sec:applications} presents applications to 
quantum computing, demonstrating how abstract computational systems can encode representations of finite-dimensional quantum gates, including for qubits and qudits. Finally, Section~\ref{sec:conclusion} concludes with a discussion of implications and future directions.

\section{Preliminaries}\label{sec:prelim}

In this Section we provide preliminary materials from Ref.~\cite{barbour1992extremal}. Sometimes we use a modified notation compared to the cited work.

\begin{definition}[Neighborhood of a character \cite{barbour1992extremal}]\label{defn:neigh_max}
    For given position ($i$) and radius ($m$) the neighborhood of the character at position ($i$) with radius ($m$) is the substring from position $i-m$ to $i+m$ considering cyclic topology, including the letter at position ($i$). If the string length is $N$, the maximum value of $m$ is as follows:

    \begin{equation}
        m_{\text{max}} = \begin{cases}
            N/2-1 & N \text{ is even}\\
            (N-1)/2 & N \text{ is odd}
        \end{cases}
    \end{equation}
\end{definition}

\begin{definition}[String isomorphism \cite{barbour1992extremal}]
    Given two strings they are deemed isomorphic if they are the same or one is the mirror image of the other.
\end{definition}

\begin{definition}[Relative indifference \cite{barbour1992extremal}]
    Relative indifference is calculated given two distinct positions in a string. It is the smallest radius of neighborhoods of these characters such that the neighborhoods are \emph{not} isomorphic. We denote the relative indifference between characters at positions $i,j$ as $r_{ij}$.
\end{definition}

\begin{definition}[Leibnizian strings \cite{barbour1992extremal}]
    A string is \emph{Leibnizian} if all character positions are mutually 
    distinguishable by their neighborhoods. Specifically, for every pair of 
    distinct positions $i \neq j$, there exists some radius $m \leq m_{\max}$ 
    such that the neighborhoods of radius $m$ at positions $i$ and $j$ are 
    non-isomorphic. A string is \emph{not} Leibnizian if two or more positions 
    have isomorphic neighborhoods at all radii.
\end{definition}

\begin{example}
    Consider the cyclic string ``AAABBB''. The first and the third A's share the same neighborhood. This is also the case of the first and the third B's. Hence these positions are 
    indistinguishable, so ``AAABBB'' is \emph{not} Leibnizian.
\end{example}

\begin{example}
    In contrast, consider ``AABABB''. At the maximum radius $m=2$ (in order to determine whether a word is Leibnizian it is enough the check the maximum radius neighborhoods):
    \begin{itemize}
        \item Position 1 (A): neighborhood = ``BBAAB''
        \item Position 2 (A): neighborhood = ``BAABA''
        \item Position 3 (B): neighborhood = ``AABAB''
        \item Position 4 (A): neighborhood = ``ABABB''
        \item Position 5 (B): neighborhood = ``BABBA''
        \item Position 6 (B): neighborhood = ``ABBAA''
    \end{itemize}
    All six neighborhoods are distinct, so ``AABABB'' \emph{is} Leibnizian. The subtle point here is that although the entire  string is cyclic, the neighborhoods are \emph{not}.
\end{example}

The motivation for the definition of Leibnizian strings comes from the monadology \cite{monadology} of Leibniz. In his terms \cite{monadology}:

\begin{quote}
  The monad [\ldots] is nothing else than a simple substance, which goes to make up composites; by simple we mean without parts.
\end{quote}

Due to the identity of indiscernibles principle of Leibniz \cite{sep-identity-indiscernible} if two objects share the same properties then they are the same. For character strings, the properties are the neighborhoods of each position in a string. For that purpose, for Leibnizian strings we require that each neighborhood of characters is distinguished from the neighborhood of other characters. Thus, Leibnizian strings can model physical systems interpreted as consisting of monads. Monads themselves can be seen as `metaphysical atoms'. Leibnizian strings will be of prime importance in our work.

\begin{definition}[Absolute indifference \cite{barbour1992extremal}]
    Absolute indifference is calculated for a single position in a string. It is the maximum of relative indifference values between the given position ($i$) and other positions in the string. It is expressed as $a_i = \max\{r_{ij} | j \neq i\}$.
\end{definition}

Next, we discuss the `BSD Variety', which was introduced in \cite{barbour1992extremal} as a complexity measure of a string of characters. The term `BSD' refers to the names ``Barbour-Smolin-Deutsch'', and was coined in this form in \cite{dundar2024case}  (since the original paper by Barbour and Smolin  \cite{barbour1992extremal} cites a private communication by Deutsch, dated 1989, as the suggestion for that specific variety function -  Ref.~[30] in \cite{barbour1992extremal}).

\begin{definition}[BSD Variety \cite{barbour1992extremal}]
    Given a Leibnizian string the BSD variety is calculated as $\variety = \sum_i 1/a_i$.
\end{definition}

\begin{example}
    The string ``AABABB'' is Leibnizian. The absolute indifference values of characters are $2,2,1,1,2,2$. Hence, its BSD variety is $4$.
\end{example}

This definition is somewhat different from the one recently used by Smolin in Refs.~\cite{smolin2021, smolin2022}. In our study, since we consider only the Leibnizian strings, we set the variety value of a \emph{non}-Leibnizian string to zero. However, in more general cases, one may sum over $a_i \neq 0$ and neglect the $a_i = 0$ terms. The case $a_i = 0$ arises when the full neighborhood of a character in position $i$ is the same as that of some other character. In this way,  one may also say that $a_i \to \infty$ and these terms do not contribute to the sum to calculate the BSD variety. However, in our algorithm, we set $a_i = 0$ in non-Leibnizian strings to store a numerical value on the computer.

\section{Information-Theoretic Aspects of Leibnizian Strings \& their Enumeration}\label{sec:Leibnizian_strings_info_theory}

In this Section, we are concerned with information-theory related concepts for Leibnizian strings and the number of Leibnizian strings.

\subsection{Information-Theoretic Considerations}

A cyclic character string is deemed Leibnizian if no two positions exhibit isomorphic local neighborhoods even when the whole string is covered. Let us give two examples: ``AAABBB'' is a non-Leibnizian string, whereas ``AABABB'' is Leibnizian. This definition prevents redundant symbolic structures, ensuring a meaningful interpretation of characters as \emph{monads} \cite{monadology} in Leibnizian terms, via the principle of identity of indiscernibles~\cite{sep-identity-indiscernible}. Because every neighborhood is unique in Leibnizian strings, they have in some sense more information than non-Leibnizian strings. For a non-Leibnizian string, about half of the string is redundant. For a detailed account of being a Leibnizian string and the calculation of (BSD) variety, please refer to \cite{barbour1992extremal,dundar2024case}. Readers may refer to \cite{barbour2003deep} as a complementary reading on Leibnizian ideas. Let us begin by defining the Shannon entropy.

\begin{definition}[Shannon entropy of a string of characters \cite{shannon_entropy}]
    For a string of length $n$ where letters are from an alphabet $\{\lambda_i|i=1,\cdots,\nu\}$ where each letter has probability $p_i = \# \lambda_i/n$ of occurrence the Shannon entropy is defined as the sum $H = - \sum_{i = 1}^\nu p_i\log_2 p_i$.
\end{definition}

By changing the base of logarithm in the definition above, one can obtain information in other untis, such as \emph{nats} for natural logarithm. As is seen, the Shannon entropy is concerned only with the frequency of letters occurring in a string. Hence, it \emph{cannot} reflect the complex structure of Leibnizian strings in terms of local neighborhoods of characters. On the other hand, the \emph{conditional} Shannon entropy (see Ref.~\cite{gray_entropy} for a discussion) is more useful in order to define the information content of a string.

\begin{definition}[Conditional Shannon entropy \cite{gray_entropy}]
    The \emph{conditional} Shannon entropy is defined as $H(Y|X) \equiv H(X,Y) - H(X)$. The term $H(Y|X)$ is the conditional entropy of $Y$ given $X$, the term $H(X,Y)$ is the joint entropy and $H(X)$ is the Shannon entropy. 
\end{definition}

Using the definition of conditional Shannon entropy we just provided, let us define it in the context of character strings.

\begin{definition}[Conditional Shannon entropy of a string of characters]
    Let $H(X)$ be the Shannon entropy of the whole string and $H(X,Y)$ be the joint entropy. In order to calculate $H(X,Y)$, we consider a list of two consecutive letters in the string. Hence $H(X,Y)$ is the Shannon entropy of a list consisting of two-letter `characters.'
\end{definition}

We considered random Leibnizian strings of varying lengths and calculated the BSD variety and the conditional Shannon entropy. The results are provided in Figure~\ref{fig:cse_var}. As we can see, there is a correlation between these two concepts. This key observation suggests that the information content of a string is positively correlated with the BSD variety.

\begin{figure}
    \centering
    \includegraphics[width=\linewidth]{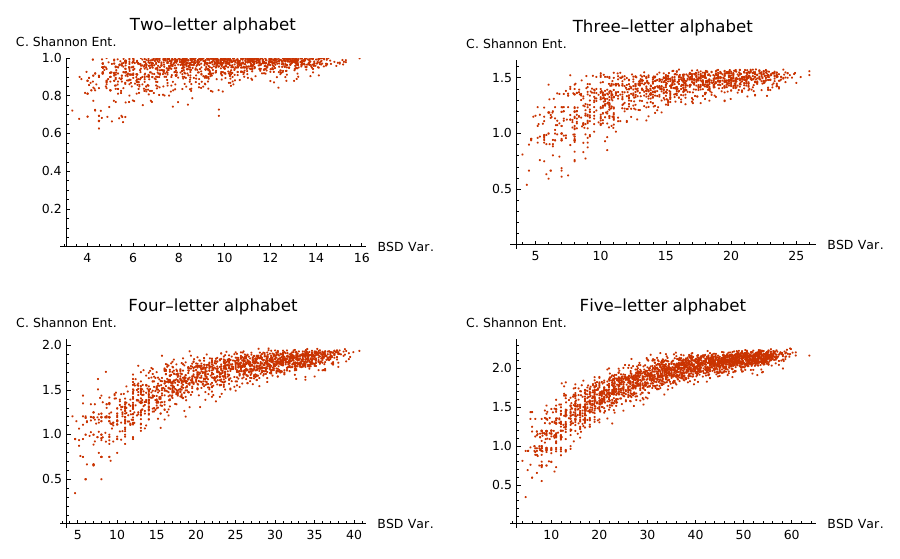}
    \caption{The plot of BSD variety vs conditional Shannon entropy for a set of Leibnizian strings for different alphabets. It is seen that BSD variety is positively correlated with the information content of words. Two-letter alphabet: 1281 words of length between 8--50. Three-letter alphabet: 1289 words of length between 8--50. Four-letter alphabet: 1704 words of length between 8--65. Five-letter alphabet: 2490 words of length between 8--90.}
    \label{fig:cse_var}
\end{figure}

\subsection{Enumeration of Leibnizian strings}

The number of Leibnizian strings is countably infinite. The \emph{countably} part is straightforward. However, we need to show that they are infinitely many. For a letter $\lambda$ we denote by $\lambda^n$ the repetition of letter $\lambda$ $n$ times. For example $A^2 = AA$ and $A^3 = AAA$ and so on. In this subsection, we consider an alphabet of letters where the number of distinct letters is $\nu$. Hence, the alphabet is the set $\{\lambda_i | i = 1,\cdots,\nu \}$.

\begin{definition}[Fractal words]
    A fractal word is a word of the form \\ $\otimes_{i=1}^n \lambda_1^i \cdots \lambda_\nu^i$ parameterized by the positive integer $n$.
\end{definition}

\begin{example}[Two letter alphabet]
    Let us consider a two letter alphabet: $\{A,B\}$. The fractal words for $n = 1,2,3$ are  ``AB'', ``ABAABB'' and ``ABAABBAAABBB'' respectively. 
\end{example}

\begin{proposition}
    \label{thm:fractal_word_leibnizian}
    All fractal words are Leibnizian for $n > 1$.
\end{proposition}

\begin{proof}
    We will prove the proposition for a two letter alphabet. Its generalization to larger alphabets will be obvious once we present our method. Now we consider a two letter alphabet $\{A,B\}$ and $n > 1$. If $n = 1$ the word is just ``AB'' and hence non-Leibnizian. Suppose now that $n > 1$. Considering the block of $A$'s in $A^iB^i$ and $A^jB^j$ for $i < j$, when we compare the neighborhoods of $A$'s in these two separate substrings, we see that they are distinguished. The reason is that the neighborhood of $A$'s in the body $j$ must contain $j$ consecutive $A$'s and similarly the neighborhood in body $i$ must contain $i$ consecutive $A$'s. Because $i \neq j$ the $A$'s in two different blocks of substrings are distinguished.

    Hence, if the string is Leibnizian, it is enough that we check the letters in the same block of the substring. For that matter, we consider the part $A^iB^i$. Let us choose the first $A$ of this part. It is enough to check the neighborhoods of the first $A$ with the last $A$. These are distinguished in $m = i$. When we consider the $r$'th $A$ it is enough that we compare it with the $(i-r+1)$'th $A$ in the block $A^iB^i$. Those, too, are distinguishable. However, we need to check that the radius where these $A$'s are distinguished is smaller than $m_{max}$ defined in Definition~\ref{defn:neigh_max}. If $i$ is even, the maximum radius is $i/2 - 1 +i = 3i/2 - 1$ and if $i$ is odd, it is $(i+1)/2 - 1 + i = (3i - 1)/2$. It is easy to verify that the radius at which two distinct $A$'s are distinguished in the substring $A^iB^i$ is less than $m_{max}$ for $n > 2$, using the fact that the maximum of $i$ is $n$. Because the string length is even, we should have $m_{max} = (n+2)(n-1)/2$. If $n = 2$ we have the string ``ABAABB'' and it is Leibnizian. One can do a similar analysis of letters `B' and make sure that the fractal words are Leibnizian. That completes the proof.
\end{proof}

\begin{corollary}[Leibnizian words are infinitely many]
    Due to Proposition~\ref{thm:fractal_word_leibnizian} for any $n > 1$ the corresponding fractal word is Leibnizian. Hence, the cardinality of the set of Leibnizian words is $\aleph_0$.
\end{corollary}

Using Proposition~\ref{thm:fractal_word_leibnizian} one can also calculate the BSD variety of fractal words.

\section{Wolfram Model and Abstract Rewriting Systems}\label{sec:nd_rewrite}

In this Section we present a brief overview of the Wolfram model and abstract rewriting  (reduction) systems. 

The Wolfram model  \cite{NKS, wolfram2020} is a computational framework based on nondeterministic rewriting systems. This model attempts to provide a description of fundamental physical phenomena starting from abstract rewriting systems. It seeks to capture ways in which simple computational programs lead to complex symbolic systems relevant to physical structure.  This approach constitutes a constructivist formalization of the foundations of physics based on formal symbolic systems. This complements ongoing efforts in the category-theoretic foundations of physics and mathematics, for instance, in relation to the synthetic geometry and homotopy type theory program \cite{schreiber2013differential, fiorenza2013higher, program2013homotopy, shulman2021homotopy}.   
The mathematical foundations of the Wolfram model in terms of higher categories and homotopy type theory have been described in \cite{arsiwalla2020homotopic, arsiwalla2021homotopies,  arsiwalla2024pregeometric}. The philosophical underpinnings of the Wolfram model can be found in \cite{rickles2023ruliology, arsiwalla2023pregeometry, rickles2025quantum}.

The general framework of the Wolfram model being grounded in formal computation, it does away with continuum structures in favor of discrete descriptions, including those referring to space, time and matter. The model posits a combinatorial description based on causal relations between events, which makes manifest the computational architecture that underlies contemporary physics. 
The archetypical structures that appear within this framework are the so-called ``multiway systems''. These are nondeterministic abstract rewriting systems equipped with a causal structure. The term ``multiway'' refers to the fact that these systems instantiate all permissible applications of rewriting rules in all possible orders, thereby leading to multiple sequences of rewriting terms that are partially ordered by binary relations. The partial order realizes a causal structure. Depending upon the precise interpretation of the terms, Wolfram model multiway systems may be realized as rewriting systems over graphs, hypergraphs, character strings, or other symbolic structures. For what follows below, here we will mainly be concerned with multiway systems of character strings.

Now let us begin with the definition of abstract rewriting (reduction) systems.

\begin{definition}[Abstract rewriting (reduction) system  \cite{bezem2003term} ]
    An \emph{abstract reduction system} (ARS) is a structure $\mathcal A = (A,\{\to_\alpha | \alpha \in I\})$ consisting of a set $A$ and a set of binary relations $\to_\alpha$ on $A$, indexed by a set $I$.  For $\alpha \in I$, the relations $\to_\alpha$ are called \emph{reduction} or \emph{rewriting} relations. 
\end{definition}

Beginning with a single element and a set of rules $\{\to_\alpha\}_\alpha$, an abstract rewriting system can generate a whole class of elements that can be obtained via these rules. The application of rules is known as \emph{reduction} \cite{bezem2003term}.

A Wolfram model is any abstract rewriting system with a designated initial state. The archetypical structure of a Wolfram model is a \textit{multiway rewriting system}. 

\begin{definition}[Multiway rewriting system \cite{NKS, wolfram2020} ]
A \emph{multiway rewriting system} of a Wolfram model is 
an abstract rewriting system $\mathcal A = (A,\{\to_\alpha | \alpha \in I\})$ with a designated initial state $a_{ini} \in A$ that generates sequences of rewriting events upon application of all possible rules from the set $\{\to_\alpha | \alpha \in I\}$, in all possible orders, to every element in each sequence (whenever rule application is valid).
\end{definition}

The Wolfram model multiway system is often represented by a multiway graph. 

\begin{definition}[Multiway graph \cite{wolfram2020} ]
A \emph{multiway graph} or \emph{multiway evolution graph} is a directed graph corresponding to the evolution of an abstract rewriting system $\mathcal A = (A,\{\to_\alpha | \alpha \in I\})$, in which the vertices of the graph refer to elements $a \in A$, and directed edges ${a \to b}$ refer to the existence of a rewriting event that transforms element $a$ to element $b$. Every multiway graph is further specified by a root vertex, which corresponds to the initial element of the rewriting sequences.   
\end{definition}

See Figure~\ref{fig:ARS1} for an illustration of a string substitution (with character strings) multiway system with initial state ``BBBAAACC'', and two rewriting rules $\text{BA}\to\text{AB}, \; \text{CB}\to\text{BC}$. In this case, the multiway graph has finitely many elements, and the rewriting process will terminate. This is because after a certain number of rewriting events, the string will be sorted. In general, a multiway system needs not terminate though.

\begin{figure}
    \centering
    \includegraphics[width=\linewidth]{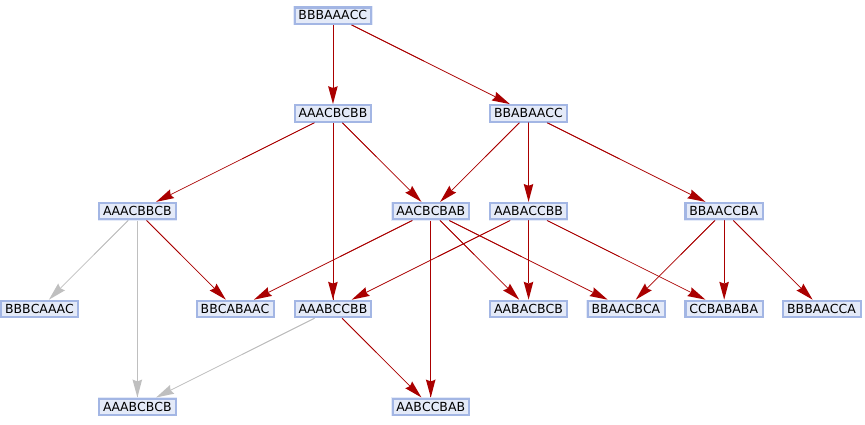}
    \caption{Initial section of the multiway system graph of string ``BBBAAACC'' with the rules: $\text{BA}\to\text{AB},\text{CB}\to\text{BC}$. The \emph{physical} multiway system that includes only the Leibnizian strings is highlighted in red.}
    \label{fig:ARS1}
\end{figure}

\section{Physical and Maximal Variety Paths}\label{sec:phys_max_var_paths}

Physical paths in the multiway system consist exclusively of Leibnizian strings. These paths evolve through a discrete substitution process. They can branch and merge within the multiway system.  

\begin{definition}[Physical paths \cite{dundar2024case}]
    A physical or a Leibnizian path in a multiway system is a path such that all the vertices it includes are Leibnizian.
\end{definition}

\begin{definition}[Maximal variety paths \cite{dundar2024case}]
    Maximal variety paths are physical paths in a multiway system of depth $d$ such that they maximize the sum of (BSD) variety values of its vertices.
\end{definition}

One needs to be careful that the maximal variety paths are \emph{not} additive. Consider maximal variety paths of depths $d$ and $d+1$. In these two cases, there may be string configurations in maximal variety paths of depth $d$ where they may not be present at depth $d+1$ maximal variety paths. On the other hand, physical (Leibnizian) paths are additive. This expression means that physical paths of depth $d+1$ for some $d$ can be obtained by continuation of physical paths of depth $d$ one step ahead.

For illustration purposes, we consider a small multiway system, see Figure~\ref{fig:mws_max_act_paths}. In order to define maximal variety paths, we consider the sum of BSD variety values on all physical paths and choose the physical paths with the maximum sum as maximal variety paths.

\begin{figure}
    \centering
    \includegraphics[width=\linewidth]{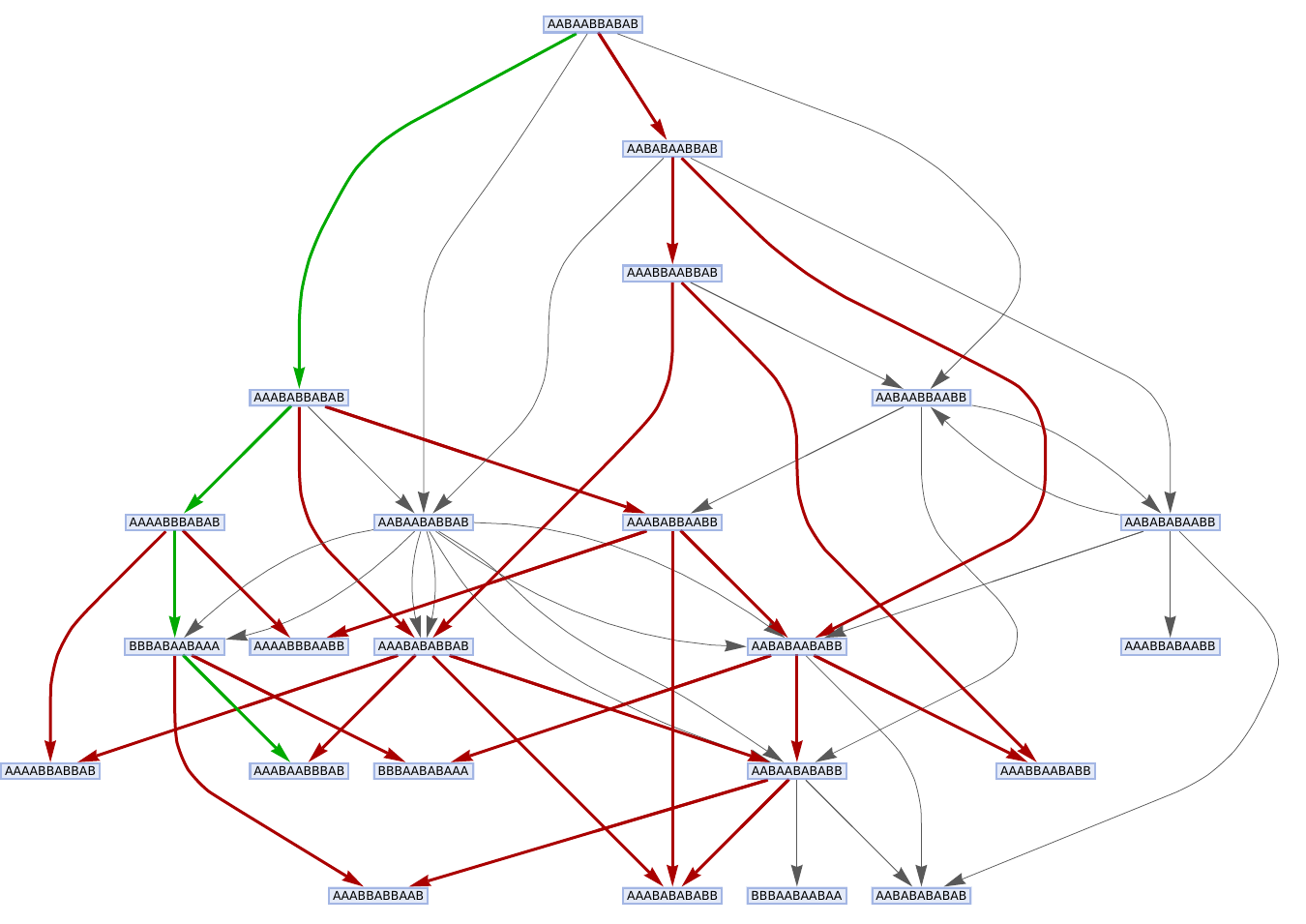} 
    \caption{Representation of initial part (depth 4) of the multiway system specified by the following: Initial state: ``AABAABBABAB''. Rewriting rule: BA$\to$AB. The paths highlighted in red are Leibnizian (physical) paths where as the path in green is the maximal variety path.}
    \label{fig:mws_max_act_paths}
\end{figure}

\section{Leibnizian Strings and Fermi-Dirac Statistics }\label{sec:fermionic}

Each monad is unique. This is dictated by Leibniz's principle of \emph{Identity of Indiscernibles}. For a detailed philosophical discussion of this principle, the esteemed reader may kindly consult \cite{sep-identity-indiscernible}. What distinguishes a monad from another one in the present string ontology is their local neighborhoods. The smallest neighborhood of a letter that distinguishes it from other neighborhoods can be said to be the \emph{view} of that monad. The term `view' is also used in \cite{smolin2021}. However, the context in \cite{smolin2021} is closer to  causal set theory. On the other hand, our use of the term `view' is closer to the discussion in  \cite{barbour1992extremal}). The view of a monad is its local neighborhood that has an absolute indifference value as its radius. Monads can be defined through their views, and a specific view can occur only once or never in  Leibnizian strings. In this sense, monads exhibit \emph{fermionic} behavior. We will soon elaborate on this.

The set of possible views consists of all possible combinations of letters in the alphabet, with radius $m=1,2,\ldots$ where the length of view for a specific radius $m$ is $2m+1$. For example, in a two letter alphabet the set of views of radius $m=1$ is 
\[  \{\text{AAA},\text{AAB},\text{ABA},\text{ABB},\text{BAA},\text{BAB},\text{BBA},\text{BBB}\}  \]
Views with radius $m=0$, that is, of only one letter, do not occur in Leibnizian strings. We denote a view of a monad with the letter $\phi$.

Now we consider a gas of Leibnizian strings of the same length, that is, $N$. We parameterize each character in a string by the $j$ parameter. In a way $j$ is the position of a character. We denote by $n_\phi$ the occupation number of view $\phi$. Because we only consider Leibnizian strings, $n_\phi$ can be zero or one. This property of Leibnizian strings will be the key to show that they exhibit fermionic behavior. Our guide in the rest of this Section will be Ref.~\cite{baierlein1999thermal} and will follow this source as regards the derivation of expected value of occupation number of a view.

The partition function of an ensemble of Leibnizian strings of length $N$ is denoted by $Z = \sum_l e^{-E_l/\kappa T} = \sum_l e^{-\beta E_l}$ where $\beta = 1/\kappa T$, the index $l$ goes over all the Leibnizian strings in the ensemble and $\kappa$ is Boltzmann's constant. For a specific view $\phi$ the expectation value of its occupation number is defined as:

\begin{equation}
    \langle n_\phi\rangle = \frac{1}{Z} \sum_l n_\phi e^{-\beta E_l}.
\end{equation}

Here we posit that the energy $E_l$ is proportional to the BSD variety of the corresponding string $\mathcal S_l$ and write $E_l = \gamma \variety(\mathcal S_l)$ where $\gamma$ is the proportionality constant that relates variety and energy. As a matter of notation, let us denote by $a_i$ the absolute indifference value of view $n_{\phi_i}$, that is, the radius of view $n_{\phi_i}$. If a Leibnizian string consists of a number of views $\{n_{\phi_i}\}_i$ where $1\leq i\leq M$, where $M$ is such that all views that can be found in a string of length $N$ are included. This will be a finite set of views. we can write its BSD variety as the sum:

\begin{equation}
    \variety(S_l) = \sum_i \frac{n_{\phi_i}}{a_i}.
\end{equation}

This observation will be important in what follows. We expand the sum to compute $Z$ in the calculation of $n_{\phi_1}$ as follows:

\begin{equation}
    \langle n_{\phi_1}\rangle = \frac{1}{Z} \sum_{n_{\phi_1}} \sum_{n_{\phi_2}} \cdots  n_{\phi_1} \exp\left[-\beta \gamma \left(\frac{n_{\phi_1}}{a_1}+\cdots +\frac{n_{\phi_M}}{a_M} \right)\right].
\end{equation}
under the condition $n_{\phi_1} + \cdots + n_{\phi_M} = N$. The triple dots denote summation operators over other possible views. The only contribution to the sum above comes from the value of $n_{\phi_1} = 1$, and we can write:

\begin{equation}
    \langle n_{\phi_1}\rangle = \frac{e^{-\beta \gamma/a_1}}{Z} \sum_{n_{\phi_2}} \cdots  \exp\left[-\beta \gamma \left(\frac{n_{\phi_2}}{a_2}+\cdots +\frac{n_{\phi_M}}{a_M} \right)\right].
\end{equation}

under the condition $n_{\phi_2} + \cdots + n_{\phi_M} = N - 1$. The idea is to write the sum above as follows:

\begin{equation}
    \sum_{n_{\phi_1}} \sum_{n_{\phi_2}} \cdots  (1 - n_{\phi_1}) \exp\left[-\beta \gamma \left(\frac{n_{\phi_1}}{a_1}+\cdots +\frac{n_{\phi_M}}{a_M} \right)\right]
\end{equation}

under the condition $n_{\phi_1} + \cdots + n_{\phi_M} = N - 1$. If we parameterize the partition function of an ensemble of Leibnizian strings of length $N$ by $Z(N)$, it will be seen that the expression above is equal to $Z(N-1) - \langle n_{\phi_1}\rangle_{N-1} Z(N-1)$ where the subscript $N-1$ in the expectation value indicates the fact that it is considered for an ensemble of strings of length $N-1$. Hence, we obtain the following useful expression:

\begin{equation}
    \langle n_{\phi_1}\rangle = \frac{e^{-\beta \gamma /a_1}}{Z(N)}\left(Z(N-1) - \langle n_{\phi_1}\rangle_{N-1} Z(N-1)\right)
\end{equation}

When considered for a single Leibnizian string if a view is found once, it will not be found in the sub-string when the corresponding letter is removed. However, we are considering an \emph{ensemble} of Leibnizian strings, so it can be found in others as well. What this means is that $\langle n_{\phi_1}\rangle_{N-1}$ is \emph{not} trivially zero. In the large $N$ limit, we can write $\langle n_{\phi_1}\rangle \approx \langle n_{\phi_1}\rangle_{N-1}$. Moreover, the \emph{chemical potential} ($\mu$) is expressed as $e^{\beta \mu} = Z(N-1)/Z(N)$. When we solve the equation above under these considerations, we obtain the final result of this Section:

\begin{equation}
    \langle n_{\phi_1}\rangle = \frac{1}{e^{\beta (\gamma/a_1 - \mu)} + 1} 
\end{equation}

Now we have shown that Leibnizian strings exhibit fermionic behavior. The important step from a Leibnizian string to this expression is to realize that monads are defined by their views. The expectation value of the occupation number is that of a view, and the corresponding energy-like term is $\gamma/a_1$. To iterate, $\gamma$ is the proportionality constant that relates the BSD variety to energy, and $a_1$ is the radius of that view, or its absolute indifference value, in other words. Similar calculations are also valid for the occupation number of other views. Hence, one can drop the subscript `$1$' in the Equation above.

\section{Multiway Path-Sums and the $S$-Matrix} \label{sec:s_matrix}

The multiway system with Leibnizian strings is a purely symbolic system. What is the physical abstraction of such a formal system? In this regard, a Leibnizian string consisting of $N$ characters can be interpreted as a $N$-fermion system, where each character (fermion) occupies a distinct neighborhood (state). This serves as the analogue of the Pauli exclusion principle for character strings. Moreover, besides the fermion-like strings, the multiway system carries a causal structure, which comes from state transitions facilitated via a given set of rewriting rules. Hence, transitions between strings (which constitute states of the multiway system) can be thought of as scattering processes between in-coming and out-going $N$-fermion states. 

To obtain the above-mentioned scattering amplitudes, we shall construct a multiway path-sum based on an action defined over a sequence of Leibnizian strings, and then use that to compute $S$-matrices corresponding to unitary transformations.

First, we define the action on Leibnizian paths in the multiway system: 
\begin{definition}[Action for Leibnizian paths (Adapted from \cite{dundar2024case})]
    Given a Leibnizian path of the form of $a_1 \to a_2 \to \cdots \to a_n$, where the $a_i$ are nodes/vertices on the multiway graph, the action along this path is defined as  
    \[ {\cal S} \, (a_1 \to a_2 \to \cdots \to a_n) = - \sum_{i=1}^{n-1} \variety (a_i) \] 
    for $n>1$.
    \label{def:action_path_sum}
\end{definition}

For example, given a path from $a\to b$, the action $\mathcal S(a\to b) = -\variety(a)$; or, given $\mathcal S(a\to b\to c) = -\variety(a) - \variety(b)$. The use of minus sign in the definition will become clear in the next subsection. 

Transition amplitudes between string states are computed using a discrete path-sum defined akin to the Feynman path integral. Path integrals over multiway systems have been discussed before in \cite{wolfram2020, gorard2020some}. The construction we use here is based on the BSD variety over a sequence of Leibnizian strings. Hence, the summand in the path-sum is $\exp(i\mathcal  S/k)$, where $\mathcal S$ is the action defined above. A real dimensionless coupling constant ($k$)  generalizes the reduced Planck’s constant ($\hbar$). In what follows, the elements of the $S$-matrix encode transition amplitudes that capture the evolution of strings under symbolic substitution rules.

Next, we define the $S$-matrix as follows:
\begin{definition}[$S$-matrix]\label{def:s_matrix}
    The elements of the S-matrix are defined as 
    \begin{equation}
       \langle \text{out} | \text{in} \rangle =\sum_{\gamma} \omega(\gamma)\, e^{i \mathcal S(\gamma)/k} 
    \label{tranamp}
    \end{equation}
    This is the transition amplitude of obtaining $\ket{\text{out}}$ beginning with $\ket{\text{in}}$, the sum is evaluated on every physical path $\gamma$ connecting `in' and `out' words, $\mathcal S(\gamma)$ is defined in Definition~\ref{def:action_path_sum} and $\omega(\gamma)$ is a suitable $\mathbb R$ valued weight function, which may be constant, for paths. Normalization of transition amplitudes can be done by choosing the weight $\omega(\gamma)$ accordingly. The transition amplitude vanishes for words that are not connected by causal paths in the multiway system since the weights $\omega$ are identically zero in that case. 
\end{definition}

Note that the use of the term `$S$-matrix' here is not intended in the sense of quantum field theory. Instead, we used the term in order to indicate  transitions between in-words and out-words corresponding to string transformations.

\begin{remark}
When we think about the path integral representation of standard non-relativistic quantum mechanics (for example, see \cite{sakurai}), the path weights are the same for every path, and it is the action whose evaluation will be different on different paths. However, in our path-sum here, the path weights are $\mathbb R$-valued free parameters. As we shall see in the case of the Hadamard gate, we also need \emph{negative} weight values. The only case where we can have all $\omega(\gamma) > 0$ on existing paths is the case of what we call the non-interacting case, meaning that there are only single paths from a definite in-word to a definite out-word. In that sense, the combinatorial path-sum defined here is a more general structure than the strict counterpart of a path integral.  
\end{remark}

The in-words and out-words are elements of a discrete space that generalizes conventional vector spaces. More specifically, the character strings we use above are elements of a ${\mathbb Z}-$module. Modules over rings are natural algebraic structures that go beyond vector spaces over fields. For the discrete setting employed here, this is realized via direct sums of abelian groups. Consider a string $w$ of length $l$, with characters $c_i$ in position $i \in \{ 1, 2, ... , l \}$
\begin{equation*}
    w = c_1 \, c_2 \, \cdots c_l
\end{equation*}
Each $c_i$ is an element of a finite character set of cardinality $n$. Upon endowing this set with modular arithmetic, it maps to the cyclic abelian group ${\mathbb Z}_n$. Then, we have:  
\begin{equation*}
    w \in \bigoplus_{i = 1}^l \, {\mathbb Z}_n
\end{equation*}
where the right-hand side is a direct sum of $l$ copies of ${\mathbb Z}_n$, which itself is a ${\mathbb Z}-$module of order $n^l$. If the multiway system consists of strings of different lengths, then the choice of $l$ is made based on the string of maximum length (which is always a finite positive integer). Furthermore, for our purposes here, we only work with character strings that are symmetric under cyclic permutations. Therefore, one has to additionally quotient the module above with cyclic permutations over strings. 

Now, over any ${\mathbb Z}-$module $M$, one can define a symmetric  ${\mathbb Z}-$bilinear complex-valued form:  
\begin{equation*}
    \langle \cdot \, , \, \cdot \rangle \, : \, M \times M \to  {\mathbb C}
\end{equation*}
This replaces the notion of an inner product over a vector space. In our case, this map is specified as in Equation~(\ref{tranamp}), and enables the computation of the transition amplitude between in-words and out-words. Note that the map in Equation~(\ref{tranamp}) is a partial map (with respect to the module of strings) since its definition makes use of causal paths of the given multiway system.


What kind of compositional properties do the $S$-matrices defined by the map in Equation~(\ref{tranamp}) admit? Since edge weights in the multiway graph are free parameters, constraints on these weights determine how successive $S$-matrices in the multiway system compose, and the ensuing associativity of the product. We will now solve these constraints to show that matrix multiplication of the kind
\begin{eqnarray}
    \sum_c \langle w_{out}| w_c\rangle \langle w_c| w_{in}\rangle =  \langle w_{out}| w_{in}\rangle
\label{matmul}
\end{eqnarray}
is satisfied as a construct on the multiway system such that the $S$-matrix obtained from paths traversing several layers of the multiway system is identical to that obtained by matrix multiplication of $S$-matrices associated to intermediate layers.

First, let us consider three successive multiway layers. See Figure~\ref{fig:3_layer_system}. The incoming word is denoted by the index $i$, the outgoing word is denoted by the index $j$, and the word in the middle layer is denoted by $c$. We  compute the $j,i$ component of the $S$-matrix. For this, the index $c$ runs over all words in the intermediate layer. For this three-layer setup, Equation~(\ref{matmul}) yields the following:

\begin{equation}
    \sum_c \omega_{jc}\omega_{ci} e^{i (S_{ci} + S_{jc}) / k} = \sum_c \omega_{jci} e^{i(S_{ci} + S_{jc})/k} 
\label{mat2}
\end{equation}
where we have used the notation $\omega_{ci} = \omega(\gamma_{ci})$, $\omega_{jc} = \omega(\gamma_{jc})$ and $\omega_{jci} = \omega(\gamma_{jci})$. Likewise, the subscripts appearing in terms referring to the action, denote paths over which the action is evaluated. A straightforward solution to Equation~(\ref{mat2}) happens to be

\begin{eqnarray}
 \omega_{jc} \; \omega_{ci} =  \omega_{jci}
\end{eqnarray}

That is, the weight of a composite path is simply the product of multiway edge weights (viewed as irreducible paths). This solution ensures that  
the $S$-matrix obtained from paths traversing several layers of the multiway system is identical to that obtained by matrix multiplication of individual $S$-matrices constructed from each of the intermediate layers. That is, matrix multiplication is now guaranteed as a non-trivial composition obtained exclusively from operations on the multiway system.  
 
Furthermore, when more than one intermediate level is involved, where each level is parameterized by index $c_i$ for $i = 1,\ldots,n$, the above solution generalizes to 
\begin{eqnarray}
\omega(\gamma) = \omega_{j c_n \ldots c_1 i} = \omega_{j c_n}\omega_{c_n c_{n-1}}\cdots \omega_{c_2 c_1} \omega_{c_1 i}    \end{eqnarray}

The above also ensures that the composition of $S$-matrices is associative. 

\begin{figure}
    \centering
    \includegraphics[width=0.60\linewidth]{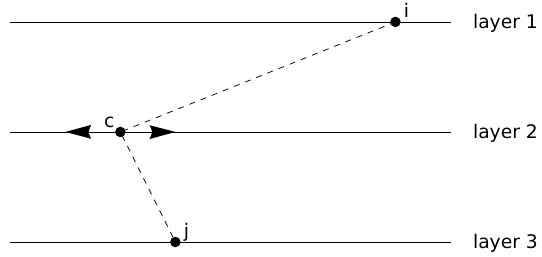}
    \caption{An illustration of a three-layer system. An in-word is denoted with $i$ and an out-word with $j$. The elements of the second layer are denoted with labels $c$, which one sums over.}
    \label{fig:3_layer_system}
\end{figure}

Two remarks are in order here: 
\begin{remark}
Regarding the matrix entries that constitute the $S$-matrix, the action $\mathcal S(\gamma)$ along a path $\gamma$ is fully determined by the words / strings that specify the path $\gamma$. However, the edge weights and, correspondingly, the path weights $\omega(\gamma)$ are free parameters. This yields a family of $S$-matrices (not necessarily unitary). The multiway system in fact encodes multiple $S$-matrices  between any fixed set of in-words and out-words. A specific choice of a $S$-matrix can be extracted only after fixing the edge weights, which is what we shall show in the next section.   
\end{remark}

\begin{remark}
If there are two different multiway systems that share the same causal structure (up to, say, some order of evolution steps), then their respective $S$-matrices, corresponding to the same multiway layers with identical connectivity, will be related in a specific way. This is because $S$-matrices constructed from two successive layers of a multiway system  decompose into a product of the matrix of path weights $\omega_{ji}$ and a diagonal matrix with entries $e^{i\mathcal S_i/k}$, where $j$ runs over all in-words 
\begin{equation*}
 \langle w_{out}| w_{in}\rangle_{ji} =   \omega_{ji} \; (e^{i\mathcal  S_i/k})_{ii}
\end{equation*}
(no Einstein summation implied above). In this case, the form of the matrix of weights $\omega_{ji}$ is the same for both $S$-matrices, meaning that non-zero entries are in the same positions of the weight matrices. Only the diagonal matrix differs with entries of the form $e^{i \mathcal S_i/k}$. Furthermore, successive $S$-matrices compose via matrix multiplication. Hence, respective $S$-matrices,  associated with multiway systems having an identical causal structure, will have a similar structural form. 
\end{remark}

\subsection{Two-Dimensional $S$-Matrices: Non-Interacting}

We begin with the construction of the simplest type of $S$-matrices: 2D and without interactions. For this we consider two in-words and two out-words. In the non-interacting case, an in-word connects to one and only one out-word. For instance, when the connectivity graph is $\{ 1_{in} \to 1_{out}, \; 2_{in} \to 2_{out} \}$ (the case $\{ 1_{in} \to 2_{out}, \; 2_{in} \to 1_{out} \}$ is also non-interacting, as we will see at the end of this subsection), the $S$-matrix takes the form:

\begin{equation}
    U = \begin{pmatrix}
        \omega_{11} e^{i\mathcal S_1/k} & 0\\ 0 & \omega_{22} e^{i\mathcal S_2/k}
    \end{pmatrix}.
\end{equation}

where $\omega$ are real-valued weights. Unitarity requires that  $\omega_{11}^2 = \omega_{22}^2 = 1$. We choose the values as one. Hence the $S$-matrix for non-interacting case becomes:

\begin{equation}
    U = \begin{pmatrix}
        e^{i\mathcal S_1/k} & 0\\ 0 & e^{i\mathcal S_2/k}    \label{eq:S_matrix_ni}
    \end{pmatrix}.
\end{equation}

Expressing this $S$-matrix as a time-evolution operator of a discrete-time system $U = \exp(-i H \Delta t)$ yields the Hamiltonian:

\begin{equation}
    H = \frac{1}{\Delta t} \begin{pmatrix}
        - \mathcal S_1/k & 0\\ 0 & -\mathcal S_2/k
    \end{pmatrix}.
\end{equation}

Notice that in order to make the Hamiltonian positive definite, we need to define the action calculated on the words $1$ and $2$ as the negative of the BSD variety of those words. This motivates the use of variety as a measure of potential energy. In \cite{smolin2021,smolin2022} Smolin also uses  variety as a form of potential energy. 

Up to a global phase factor, the general form of a 2D non-interacting $S$-matrix can be written in one of the following forms:

\begin{equation}
    U = \begin{pmatrix}
        1 & 0\\ 0 & e^{i\alpha}
    \end{pmatrix},\quad \text{or}\quad
    U = \begin{pmatrix}
        0 & 1\\ e^{i\alpha} & 0
    \end{pmatrix}
\end{equation}

for some suitable values of the phase $\alpha$. The second one above is obtained from the first one if we change the order of in-words, that is, if the connectivity data is $\{ 1_{in} \to 2_{out}, \; 2_{in} \to 1_{out} \}$. When $\alpha = \pi/4$ one obtains a realization of the $\pi/8$-gate.

\subsection{Two-Dimensional $S$-Matrices:   Interacting}

Now, we discuss 2D \emph{interacting} $S$-matrices. We add edges to the multiway graph that constitute interactions between words. Such as, between in-1 and out-2 or in-2 and out-1 words. However, there is a subtle detail here. If the interaction is one-sided, for example, if in-1 and out-2 are connected, but in-2 and out-1 are not connected, or vice versa, the off-diagonal elements must vanish. The main reason is the transpose operation in computing the  inverse of unitary matrices. More generally, for an interacting system of $N$ in-words and $N$ out-words, the $i,j$ element of a $S$-matrix, which is also unitary, will be non-zero if and only if the $j,i$ element is non-zero. That is,  the interaction has to be mutual at the level of the multiway system to ensure unitarity: if in-$j$ and out-$i$ are path-connected, then in-$i$ and out-$j$ are also path-connected.

When we consider an interacting 2D $S$-matrix,  every in-word is connected to every out-word. The $S$-matrix then reads as follows:

\begin{equation}
    U = \begin{pmatrix}
        \omega_{11} e^{i\mathcal S_1/k} & \omega_{12} e^{i\mathcal S_2/k}\\ \omega_{21} e^{i\mathcal S_1/k} & \omega_{22} e^{i\mathcal S_2/k}
    \end{pmatrix}.
\end{equation}

Imposing unitarity gives:

\begin{equation}
    U = \begin{pmatrix}
        \sqrt{1-\lambda^2} e^{i\mathcal S_1/k} & \lambda e^{i\mathcal S_2/k}\\ -\lambda e^{i\mathcal S_1/k} & \sqrt{1-\lambda^2} e^{i\mathcal S_2/k}\label{eq:S_i}
    \end{pmatrix}
\end{equation}

for a real parameter $\lambda \in [-1,1]$. Here, $\lambda$ is a free parameter that controls the magnitude of interactions. See Figure~\ref{fig:2by2} for the anatomy of 2D systems.

\begin{figure}
    \centering
    \includegraphics[width=0.5\textwidth]{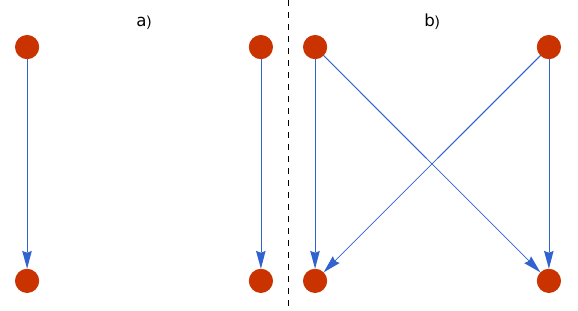}
    \caption{Anatomy of 2D systems. (a) non-interacting case, (b) interacting case. All the nodes are supposed to be Leibnizian strings. Note that there is a symmetry of permutation of words. So the case when in-1 and out-2, and in-2 and out-1 are connected is also a non-interacting set-up.}
    \label{fig:2by2}
\end{figure}

Up to a global phase factor, we can then write $U$ in Equation~(\ref{eq:S_i}) as follows:

\begin{equation}
    U = \begin{pmatrix}
        \cos(\theta) & -\sin(\theta) e^{i\alpha}\\ \sin(\theta) & \cos(\theta) e^{i\alpha}
    \end{pmatrix}
\end{equation}

\subsection{Higher Dimensional $S$-Matrices }

Let us now consider more general forms of $S$-matrices consisting of $n$ in-words and $m$ out-words. We begin with square $S$-matrices. 

In general, due to the structure of the $S$-matrix, for a square unitary matrix $U$ (that is, $n\times n$) we can write the matrix elements as follows:

\begin{equation}
    U_{ij} = \omega_{ij} e^{i \mathcal S_j/k} 
\end{equation}

Each column vector $c_i$ of $U$ consists of elements $\{\omega_{li}\}_l$ multiplied with the constant phase $e^{i \mathcal S_l/k}$. The unitarity requirement of $U$ implies $c_i^\dagger c_j = \delta_{ij}$. Hence, if we write

\begin{equation}
    U = \omega \; \text{diag}(e^{i \mathcal S_1/k},e^{i \mathcal S_2/k},\ldots,e^{i \mathcal S_n/k})    
\end{equation}

it is seen that the $\omega$ matrix is orthogonal. The general $3\times 3$ orthogonal $\omega$ matrix can be expressed as follows \cite{goldstein} (without the factor $\pm 1$):

\begin{equation}
    \omega= \pm\begin{pmatrix}
        \cos \psi \cos \phi - \cos\theta \sin\phi \sin\psi & \cos\psi\sin\phi+\cos\theta\cos\phi\sin\psi & \sin\psi\sin\theta \\
        -\sin\psi\cos\phi-\cos\theta\sin\phi\cos\psi & -\sin\psi\sin\phi + \cos\theta\cos\phi\cos\psi & \cos\psi\sin\theta \\
        \sin\theta\sin\phi & -\sin\theta\cos\phi & \cos\theta 
    \end{pmatrix}
\end{equation}

where we used the Euler angles to denote a generic element of $SO(3)$ and added the factor $\pm 1$, in order to construct elements of $O(3)$. Hence, the generic $S$-matrix we can represent is expressed as follows:

\begin{equation}
    U_{3\times 3} = \omega\,  \text{diag}(e^{i \mathcal S_1/k},e^{i \mathcal S_2/k},e^{i \mathcal S_3/k})
\end{equation}

In particular, the $4\times 4$ $S$-matrices have the following forms:

\begin{eqnarray}  
    U_{4\times 4} &= \begin{pmatrix}
        \omega_{11} e^{i \mathcal S_1/k} & \omega_{12} e^{i \mathcal S_2/k} & \omega_{13} e^{i \mathcal S_3/k} & \omega_{14} e^{i \mathcal S_4/k}\\
        \omega_{21} e^{i \mathcal S_1/k} & \omega_{22} e^{i \mathcal S_2/k} & \omega_{23} e^{i \mathcal S_3/k} & \omega_{24} e^{i \mathcal S_4/k}\\
        \omega_{31} e^{i \mathcal S_1/k} & \omega_{32} e^{i \mathcal S_2/k} & \omega_{33} e^{i \mathcal S_3/k} & \omega_{34} e^{i \mathcal S_4/k}\\
        \omega_{41} e^{i \mathcal S_1/k} & \omega_{42} e^{i \mathcal S_2/k} & \omega_{43} e^{i \mathcal S_3/k} & \omega_{44} e^{i \mathcal S_4/k}
    \end{pmatrix}
\end{eqnarray}

Here $\omega \in O(4)$ . A few definitions will be useful in what follows.

\begin{definition}[Fully interacting system]
    A fully interacting system (as a subsystem of the multiway system) is a complete bipartite graph (as a subgraph of the multiway graph). That is, all of the $n$ in-words in the designated $S$-matrix construction are connected to all of the $m$ out-words.
\end{definition}

In the presence of $n$ in-words and $m$ out-words, the dimensions of the $S$-matrix is $m\times n$. In general, the $S$-matrix will be non-square. Hence, the definition of unitarity requires more care. Note that when we define the $S$-matrix, the columns are multiplied by a pure phase function of the form $\exp(i \mathcal S_i/k)$, where $\mathcal S_i$ is proportional to the BSD variety of the $i^{th}$ in-word. When we neglect these pure phases, the remaining matrix that consists of the $\omega$'s will be orthogonal since these parameters are real valued. Hence, it is sufficient to consider orthogonal matrices and to include pure phases in the columns of these matrices.

Generalization of unitarity or orthogonality to non-square matrices is known as semi-unitary or semi-orthogonal matrices, respectively.

\begin{definition}[Semi-unitary matrix]
    If $U$ is a $m \times n$ matrix that satisfies $U^\dagger U = I_n$, we say that $U$ is a semi-unitary matrix.
\end{definition}

In the literature, the above definition can either be $U^\dagger U = I_n$ or $U U^\dagger = I_m$. However, because we consider $S$ as acting on the initial states, we have chosen the specific form of the definition. 

\begin{definition}[Semi-orthogonal matrix]
    If $U$ is a $m\times n$ matrix that satisfies $U^T U = I_n$, we say that $U$ is a semi-orthogonal matrix.
\end{definition}

A general form of the $S$-matrix (possibly rectangular) will have entries as follows. Consider a two layer system, suppose there are $n$ in-words and $m$ out-words. If the $i^{th}$ in-word is connected to $j^{th}$ out-word, then the $j, i$ element of the $S$-matrix is of the form $\omega_{ji} e^{i \mathcal S_i/k}$, where $\omega_{ji}$ is a free $\mathbb R$-valued parameter. If these words are not connected, then the $j, i$ element of the $S$-matrix is zero. The analogous pattern holds for constructions using more than two layers of the multiway system (in that case one has to consider path weights and the action evaluated along a path).

As an example, let us consider a two by three system. This system has two in-words and three out-words. The $S$-matrix will be of the following form:

\begin{equation}
    U = \begin{pmatrix}
        \omega_{11} e^{i \mathcal S_1/k} & \omega_{12} e^{i \mathcal S_2/k}\\
        \omega_{21} e^{i \mathcal S_1/k} & \omega_{22} e^{i \mathcal S_2/k}\\
        \omega_{31} e^{i \mathcal S_1/k} & \omega_{32} e^{i \mathcal S_2/k}
    \end{pmatrix}.
\end{equation}

Equation $U^\dagger U = I_2$ can be solved since there are enough parameters. Hence, in this case we can find a semi-unitary $S$-matrix. We now consider a three by two system. This system consists of three in-words and two out-words. The $S$-matrix is as follows:

\begin{equation}
    U = \begin{pmatrix}
        \omega_{11} e^{i \mathcal S_1/k} & \omega_{12} e^{i \mathcal S_2/k} & \omega_{13} e^{i \mathcal S_3/k}\\
        \omega_{21} e^{i \mathcal S_1/k} & \omega_{22} e^{i \mathcal S_2/k} & \omega_{23} e^{i \mathcal S_3/k}
    \end{pmatrix}.
\end{equation}

If this matrix is to be semi-unitary, then the matrix equation $U^\dagger U = I_3$ should be satisfied, where $I_3$ is the $3 \times 3$ identity matrix. However, there are no solutions for this equation. That is, unitarity does not hold in this case, no matter what choice of parameter values are taken.

\subsection{A Comment on Restoring Unitarity }

Typically,  unitarity is associated to a system being closed. Whereas, when unitarity breaks down, it is suggestive of an open system. As we shall see, if there are unconnected words, we cannot have a unitary $S$-matrix. Also, in cases where the number of words decreases in the next layer, we might not have a unitary $S$-matrix. On the other hand, if the number of words in a layer stays the same or increases, we can always have (semi)-unitary $S$-matrices. 

Let us now discuss how the unitarity of $S$-matrices can be restored by extending the system. As an example, consider the multiway system represented in Figure~\ref{fig:non_interacting_4_2} that relates two qubits to a single qubit. That means we have four in-words and two out-words.

The quantum gate that arises from the multiway system in Figure~\ref{fig:non_interacting_4_2} can map $\ket{00}\to \ket{0}$ and $\ket{11}\to e^{i\alpha}\ket{1}$, and $\ket{01}\to 0, \ket{10}\to 0$. Hence, the $S$-matrix cannot be unitary for such a multiway system. That process describes a \emph{non}-interacting system. If we would like to associate an extension to this system, it can be thought of as consisting of two in-words that do not interact with the rest of the system in any way. In cases where there is at least one unconnected in-word or out-word, the system on its own does not yield a unitary $S$-matrix. That would mean either at least one zero column or row.

\begin{figure}
    \centering
    \includegraphics[width=0.6\linewidth]{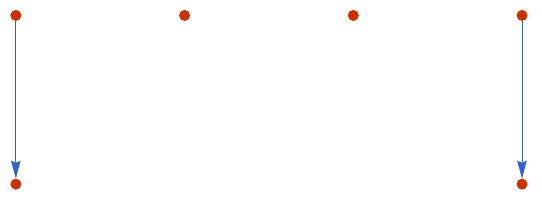}
    \caption{A non-interacting system with four in-words and two out-words.}
    \label{fig:non_interacting_4_2}
\end{figure}

In Figure~\ref{fig:222} we see an example of a situation where a rewriting rule does not match at least one word in a multiway system. When we consider the two consecutive layers, we can write $U = U_2 U_1$ where $U_1$ is unitary. However, $U_2$ cannot be a unitary matrix, since the second column is identically zero. The zero column of $U_2$ removes all the contribution from the second word of the second layer to the composite $S$-matrix $U$ due to the structure of matrix multiplication. 

\begin{figure}[h]
    \centering
    \includegraphics[width=0.35\linewidth]{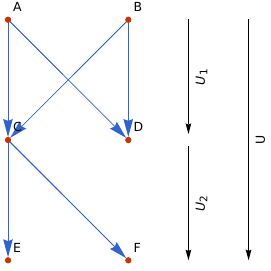}
    \caption{An example for a multiway system where a re-write rule does not match at least one word. Every word in the illustration is supposed to be Leibnizian.}
    \label{fig:222}
\end{figure}

For such a case, we can consider a transition matrix from the first layer to the third layer. Because there is only one path that connects a given in-word to any out-word, the $S$-matrix will be as follows:

\begin{equation}
    U = \begin{pmatrix}
        \Omega_{11} e^{i (\mathcal S_A+\mathcal S_C)/k} & \Omega_{12} e^{i (\mathcal S_B+\mathcal S_C)/k}\\
        \Omega_{21} e^{i (\mathcal S_A+\mathcal S_C)/k} & \Omega_{22} e^{i (\mathcal S_B+\mathcal S_C)/k}
    \end{pmatrix}\label{eq:S_matrix_nonadjacent}
\end{equation}

where we used $\Omega_{ij}$ instead of $\omega_{ij}$ in order to indicate that we have an $S$-matrix that relates non-adjacent layers. In terms of complex phases, this form is equivalent to a $2\times 2$ interacting system. Hence, even though there is a word in the second layer where a rewriting does not match, one can find a unitary transition matrix that connects the first and third layers.

On the other hand, there can be no \emph{completely zero rows} in the $S$-matrix which is because when we generate the multiway system an out-word should have been found out via application of rewriting rule that is applied to at least one of the in-words. From now on we suppose there are no completely zero \emph{columns} of $S$-matrix in addition to existence of no zero rows. Suppose now that there are $n$ in-words and $m$ out-words and there are no unconnected words considering these adjacent layers in the multiway system. The shape of $S$-matrix is $m\times n$. Hence, the matrix equation $U^\dagger U = I_n$ consists of $n^2$ individual equations. However, the number of independent equations is $n(n+1)/2$. On the other hand, the number of free parameters in $U$ is $mn$. Hence, the matrix equation can be solved if $mn \geq n(n+1)/2 \implies 2m \geq n+1 \implies m\geq (n+1)/2$. So, if the number of in-words ($m$) satisfies this inequality we can find $U$ such that $U^\dagger U= I_n$. The number of free parameters of $U$ \emph{after} the condition of (semi)-unitarity is found as follows:

\begin{equation}
    mn - \frac{n(n+1)}{2} = \frac n 2 (2m-n-1).
\end{equation}

For the case of a square $S$-matrix where $m = n$ the number of free parameters is $n(n-1)/2$. Since $n\geq 1$ the number of free parameters is a non-negative integer. However, if there are quite less out-words than in-words, such that the inequality is \emph{not} satisfied, in other words, if the multiway system \emph{quickly} decreases in number of words in the next layer, equation $U^\dagger U = I_n$ cannot be solved due to lack of free parameters. These free parameters are path weights ($\omega$) that we have discussed above.

The most economic way to solve the problem is to introduce $\Delta m$ auxiliary words on the second layer such that $2 (m+\Delta m)\geq n + 1$. The minimum number of $\Delta m$ hence satisfies $\Delta m \geq (n+1)/2 - m$. If $n$ is even $\Delta m_{\text{even}} = n/2 +1-m$ and if $n$ is odd $\Delta m_{\text{odd}} = (n+1)/2 - m$.

If $U$ denotes the $S$-matrix that relates in-words and out-words and $U_e$ is a suitable matrix that has the shape $\Delta m \times n$ we can write the \emph{composite} $U_c$ matrix as follows:

\begin{equation}
    U_c \equiv \begin{pmatrix}
        U\\ U_e
    \end{pmatrix}\label{eq:S_composite}
\end{equation}

where $U_e$ is considered as a part of the composite $U_c$ matrix that provides a means to capture information loss. The matrix $U_c$ has enough parameters that it can be put in to a (semi)-unitary matrix form. When we write $U_c^\dagger U_c = I_n$ we see that the parameters of $U$ and $U_e$ are not independent of each other. $U$ has $mn$ free parameters and $U_e$ has $\Delta m\, n$ free parameters, and in total there are enough free parameters in $U_c$, which guaranties that we can solve $U_c^\dagger U_c = I_n$. If we consider $U_e$ as a function of $U$, that means that the parameters of $U_e$ can be solved. One can also do the reverse and try to solve $U$ as a function of $U_e$. Hence, the unitarity of $U_c$ imposes restrictions on the connection structure of the auxiliary words in the multiway system.

In the next layer, we should consider $\Delta m$ auxiliary words. In other words, the rewriting rule should not match these auxiliary words. That means that some columns of the $S$-matrix for the next layer are zero, and the existence of auxiliary words has no effect on the $S$-matrix when we consider one more layer and matrix product of $S$-matrices. This is expected because in Definition~\ref{def:s_matrix} we consider Leibnizian paths that connect in and out-words. Hence, words in intermediate layers where the rewriting rule does not match are not included in the calculation of the matrix.

In general, the multiway system will produce many $S$-matrices that are not unitary. This may be useful for considering operators related to open systems, including non-unitary ZX diagrams. On the other hand, when it is necessary to work with unitary transition matrices, then the prescription for restoring unitarity described above can be implemented: namely, one  considers composite $U_c$ matrices as in Equation~(\ref{eq:S_composite}) and specifies the extension of the system at each step where $2m < n+1$ applies.


\section{Applications}
\label{sec:applications}

In this Section we consider various applications of our system to quantum gates. At the end of the Section we comment on quantum circuits for qudits as a further application area.

\subsection{Single Qubit: Hadamard and $\pi / 8$ Gates}  
\label{subsec:single_qubit}

In this subsection, we consider a fully interacting $2\times 2$ system and in particular the one depicted in Figure~\ref{fig:fi_22}.

\begin{figure}
    \centering
    \includegraphics[width=.8 \linewidth]{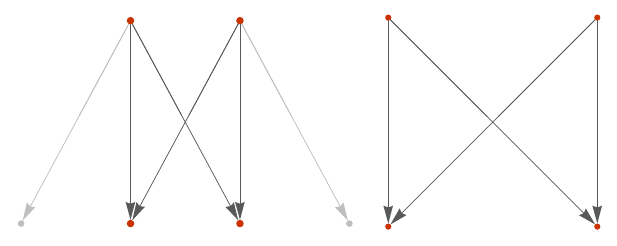}
    \caption{A single step evolution of the multiway system. The in-words are ``AABBDCABABDC'' and ``ABABCDABABDC'' and the rewriting rule is BA$\to$AB, DC$\to$CD. On the left is the full multiway system and on the right we have its physical subset. The BSD variety of in-words are 10 and those of out-words are 12.}
    \label{fig:fi_22}
\end{figure}

The $S$-matrix of this system is proportional to an orthogonal matrix, because we can factor out the phase terms which appear the same in every column of the $S$-matrix. Up to permutation of in and out-words the $S$-matrix is found in the following form

\begin{equation}
    \begin{pmatrix}
        \cos(\theta) & -\sin(\theta)\\ \sin(\theta) & \cos(\theta)
    \end{pmatrix}.
\end{equation}

If we choose $\theta = -\pi/4$ we obtain:

\begin{equation}
    U = \frac{1}{\sqrt{2}}\begin{pmatrix}
        1 & 1\\ -1 & 1
    \end{pmatrix}.
\end{equation}

We see that $U = H \sigma_x$. Due to permutations of in and out words, without loss of generality, we can write $U \sim H$, which means that we can represent the Hadamard gate. If we consider a non-interacting system of two in-words and two out-words, we can write down the form of the $S$-matrix as follows:

\begin{equation}
    U = \begin{pmatrix}
        \omega_{11} e^{i\mathcal S_1/k} & 0\\
        0 & \omega_{22} e^{i\mathcal S_2/k}
    \end{pmatrix}.
\end{equation}

If we choose $\omega_{11} = \omega_{22} = 1$ up to a global factor, we can write $U$ as follows:

\begin{equation}
    U = \begin{pmatrix}
        1 & 0\\
        0 & e^{i \Delta \mathcal S/k}
    \end{pmatrix}
\end{equation}

where $\Delta S \equiv S_2 - S_1$. If $\Delta S = \pi/4$ then our system can represent the $\pi/8$ gate.

\subsection{Two Qubits}\label{subsec:two_qubits_spin}

In this part, we will show that one can encode the CNOT and the SWAP gates using multiway systems. These two gates correspond to non-interacting systems and can be represented by permuting the order of appropriate words in a specific layer.

\subsubsection{The CNOT Gate}

In this subsection, we consider the two-layer multiway system illustrated in Figure~\ref{fig:mw_cnot}. Hence, we can write the general form of the $S$-matrix as follows:

\begin{figure}
    \centering
    \includegraphics[width=0.5\linewidth]{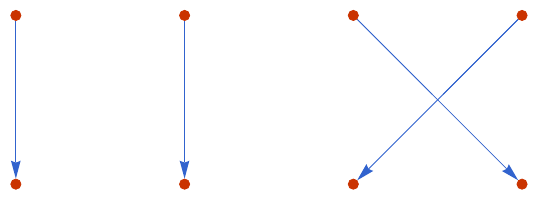}
    \caption{A two-layer multiway system. This is a non-interacting system that represents the CNOT gate. We suppose that the variety value of each in-word is the same.}
    \label{fig:mw_cnot}
\end{figure}

\begin{equation}
    U = \begin{pmatrix}
        \omega_{11} e^{i \mathcal S_1/k} & 0 & 0 & 0\\
        0 & \omega_{22} e^{i \mathcal S_2/k} & 0 & 0\\
        0 & 0 & 0 & \omega_{34} e^{i \mathcal S_4/k}\\
        0 & 0 & \omega_{43} e^{i \mathcal S_3/k} & 0
    \end{pmatrix} = e^{i \mathcal S/k}
    \begin{pmatrix}
        \omega_{11} & 0 & 0 & 0\\
        0 & \omega_{22} & 0 & 0\\
        0 & 0 & 0 & \omega_{34} \\
        0 & 0 & \omega_{43}  & 0
    \end{pmatrix}
\end{equation}

where we supposed that the variety value of each in-word is the same: $S$. One of the solutions for this form of the $S$-matrix is known as the CNOT gate:

\begin{equation}
    U = e^{i \mathcal S/k}
    \begin{pmatrix}
        1 & 0 & 0 & 0\\
        0 & 1 & 0 & 0\\
        0 & 0 & 0 & 1 \\
        0 & 0 & 1  & 0
    \end{pmatrix}
\end{equation}

Here, the global phase $e^{i\mathcal S/k}$ has no physical role, and the multiway system can encode the CNOT gate.

\subsubsection{The SWAP Gate}

In this subsection, we consider the system depicted in Figure~\ref{fig:qc3}. The BSD variety of strings in the physical multiway system is the same, namely 8. Two layers in the physical multiway system have the same structure. For that purpose, we consider the first layer. Its $S$-matrix is of the following form where: 

\begin{figure}[h]
    \centering
    \includegraphics[width=\linewidth]{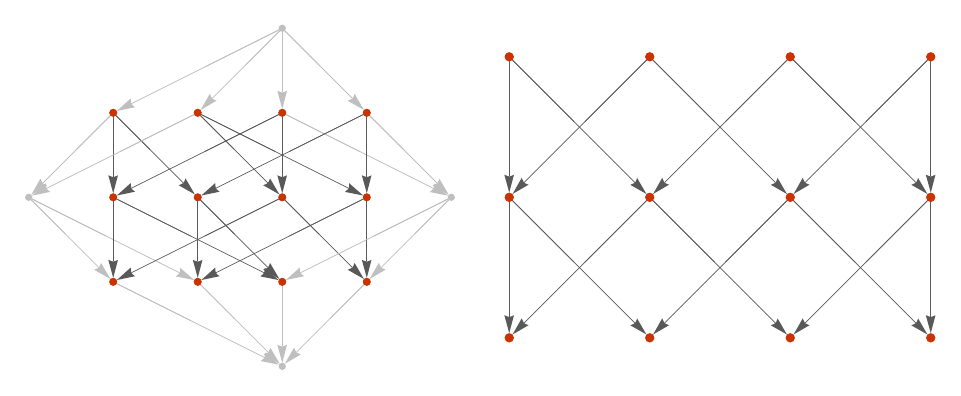}
    \caption{The full multiway system on the left hand side and its physical subset on the right hand side. The head node is ``BADCBADC'' and the rewriting rule is $\{BA\to AB, DC\to CD\}$.}
    \label{fig:qc3}
\end{figure}

\begin{equation}
    U = e^{-i V/k}\begin{pmatrix}
        \omega_{11} & \omega_{12} & 0 & 0\\
        \omega_{21} & 0 & \omega_{23} & 0\\
        0 & \omega_{32} & 0 & \omega_{34}\\
        0 & 0 & \omega_{43} & \omega_{44}
    \end{pmatrix},\label{eq:S_matrix_magnetic}
\end{equation}

Because all of the in-words have the same BSD variety, we could factor out the exponential phase factor. One of the solutions for such a matrix with the unitarity condition is the following matrix:

\begin{equation}
    U = e^{-i V/k}\begin{pmatrix}
        1 & 0 & 0 & 0\\
        0 & 0 & 1 & 0\\
        0 & 1 & 0 & 0\\
        0 & 0 & 0 & 1
    \end{pmatrix}.
\end{equation}

Apart from the global phase factor, this matrix is known as the SWAP gate. It can be represented using three CNOT gates. The minimal multiway system that encodes the SWAP gate is just a non-interacting system similar to the case depicted in Figure~\ref{fig:mw_cnot} where instead of the third and fourth words, the order of the second and third words are interchanged among themselves. See Figure~\ref{fig:mw_swap}.

\begin{figure}
    \centering
    \includegraphics[width=0.5\linewidth]{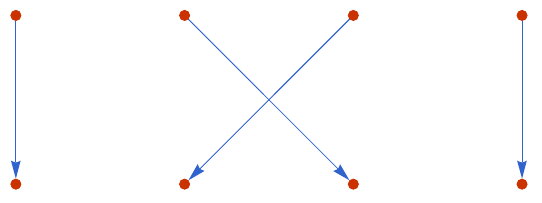}
    \caption{A minimal two-layer multiway system that can encode the SWAP gate.}
    \label{fig:mw_swap}
\end{figure}

\subsection{Qudits}

In general, our systems can be thought of as producing quantum circuits for qudits. A qudit is a generalization of two-level qubits to higher dimensions, which we define below.

\begin{definition}[Qudit]
    A qudit is a $d$-level quantum system. The basis states can be labeled as $\ket{0},\ket{1},\ldots,\ket{d-1}$. For $d=2$ one has the usual qubits.
\end{definition}

However, not all possible types of circuits are allowed. The array of character strings that correspond to components of qudits should be Leibnizian. On the other hand, a qudit can be thought of as comprising of smaller qudits. If $d = p_1^{n_1} \ldots p_a^{n_a}$ is the expression of positive integer $d$ in terms of its prime divisors $p_1,\ldots,p_a$, then we can consider this large qudits as a composite structure of $n_1$ $p_1$-qudits, $n_2$ $p_2$-qudits, and so on. An obvious example is when $d = 2^n$ for some $n$. In this case, the qudit is composed of $n$ qubits. We used this result in the example in subsection~\ref{subsec:two_qubits_spin} where we considered a two-qubit system. Another trivial example is $d = 6$. This qudit is composed of a qubit and a qutrit.

\begin{figure}
    \centering
    \includegraphics[width=0.45\linewidth]{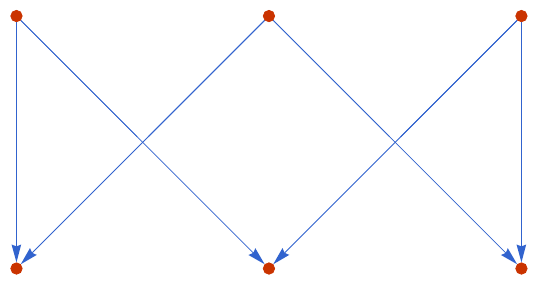}
    \caption{A portion of a (hypothetical) multiway system that corresponds to a qutrit. We suppose that the variety of each in-word is the same: $V$.}
    \label{fig:S_3_3}
\end{figure}

Let us consider a system of single \emph{qutrit}, which is a 3-level system. The configuration is depicted in Figure~\ref{fig:S_3_3}. The most general form of $S$-matrix is as follows:

\begin{equation}
    U = e^{-i V/k}\begin{pmatrix}
        \omega_{11} & \omega_{12} & 0\\
        \omega_{21} & 0 & \omega_{23}\\
        0 & \omega_{32} & \omega_{33}
    \end{pmatrix},
\end{equation}

We solved for $U^\dagger U = I_3$ and obtained (eight) solutions up to a permutation of words as follows:

\begin{equation}
    U = e^{-i V/k}\begin{pmatrix}
        \pm 1 & 0 & 0\\
        0 & 0 & \pm 1\\
        0 & \pm 1 & 0
    \end{pmatrix},\label{eq:S_matrix_form_3_3}
\end{equation}

Considering all the matrix entries as 1, this unitary corresponds to $\ket 0\to\ket 0, \ket 1\to \ket 2, \ket 2\to\ket 1$ up to a global phase. In a way, it can be seen as a SWAP gate \cite{qutrit_swap}.

In a multiway system that expands as the depth increases the dimension of the vector space of words at a certain time, the $d$ parameter of the qudit, increases. In the $d\to\infty$ limit, in this sense, we expect that our system can represent rather complex quantum systems. It would be an interesting future research direction to investigate how the dynamics of some simple quantum systems like the harmonic oscillator can be expressed within our system.

\subsection{Multiway System to Quantum Circuits}

In the above discussions concerning the form of the $S$-matrix, we have considered, most of the time, what is known as the \emph{computational basis}. These are a set of orthonormal vectors. For example $\ket 0, \ket 1$ for a qubit, $\ket{00},\ket{01},\ket{10},\ket{11}$ for a system of two qubits, and $\ket 0, \ket 1, \ket 2$ for a qutrit. In the remaining subsections of this Section, when we write the term `multiway system' we actually refer to \emph{Leibnizian} part of the multiway system where each vertex is Leibnizian.

The case when the number of elements at a certain depth in the multiway system is a power of two this system corresponds to a relation between some number of qubits. Once the (family of) $S$-matrices is determined, one can produce a (family of) quantum circuits.

If the system of in and out-words constitute a graph that is composed of disjoint subgraphs as $g = \cup_{i=1}^n g_i$ and $g_i \cap g_j = 0$ for $i\neq j$ then the $S$-matrix should be a tensor product of each $S_i$-matrices in the form of $U = \bigotimes_{i=1}^n U_i$. In ZX-diagrams these will be represented by parallel lines on which there are local operations. Having put forward the tensor product structure, we can now consider in and out-words that do not have disjoint subgraphs.

We have discussed above the form of the $S$-matrix that we can obtain in our formalism. If it is a square matrix, the $i,j^{\text{th}}$ and $j,i^{\text{th}}$ components of the $S$-matrix can be non-zero if the interaction between the words is mutual. If the number of words increases in the next step, we could still find a semi-unitary $S$-matrix, and if the number of words decreases in the next layer, we can consider a possible extension of the system to restore unitarity.

If the $S$-matrices at each level have the same dimensions (\emph{e.g.} $n\times n$) we can easily associate a quantum circuit for a given multiway system using qudits. The subtlety arises when the number of words in the next layer differs. If the number of words increases,  we can consider this situation as embedding a state vector in a larger space. If the number of words decreases, we may include an environment as a (possibly composite) qudit.

\begin{figure}
    \centering
    \includegraphics[width=\linewidth]{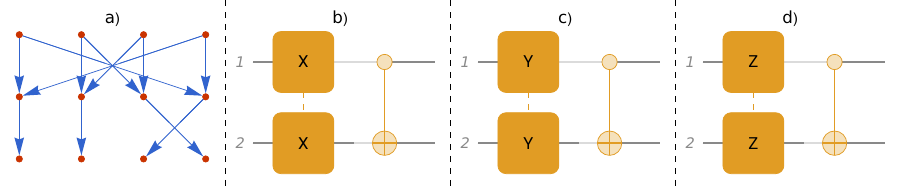}
    \caption{A three-layer multiway system (where words are ordered in the computational basis) and a few example quantum circuits it can correspond to alongside others. The time for quantum circuits flows from left to right.}
    \label{fig:mw_qc}
\end{figure}

Let us now think of a concrete example. Considering the multiway system appearing in Figure~\ref{fig:mw_qc} the forms of $S$-matrices for the first two and last two layers are as follows (we suppose that the words in each layer have the same variety value):

\begin{equation}
    U_{12} = e^{-i V_1 / k}\begin{pmatrix}
        \omega_{11} & 0 & 0 & \omega_{14}\\
        0 & \omega_{22} & \omega_{23} & 0\\
        0 & \omega_{32} & \omega_{33} & 0\\
        \omega_{41} & 0 & 0 & \omega_{44}
    \end{pmatrix},\quad 
    U_{23} = e^{-i V_2 / k}\begin{pmatrix}
        \omega_{11} & 0 & 0 & 0\\
        0 & \omega_{22} & 0 & 0\\
        0 & 0 & 0 & \omega_{34}\\
        0 & 0 & \omega_{34} & 0
    \end{pmatrix}
\end{equation}

where the parameters of $U_{12}$ and $U_{23}$ are independent from each other. According to the forms of given $S$-matrices, $U_1$ may correspond to $\sigma_n \otimes \sigma_n$ for $n = x,y,z$ or SWAP gate. $U_2$ may correspond to a CNOT gate.

There is a subtle issue about the interpretation of gates. Suppose that we have a $n$ qubit system with $S$-matrix dimension $2^n \times 2^n$. It is enough to consider the $n=2$ case to illustrate the main point. In this case, the $S$-matrix has the dimension $4\times 4$. This system can be considered to be consisting of 2 qubits or a single 4-qudit. In terms of matrix properties, there is no difference between the two. However, the \emph{interpretation} of gates differs in both cases. As a simple example, consider the $U_{23}$ matrix just above. If we consider the system as consisting of 2 qubits, this gate may correspond to a CNOT gate, where a controlled NOT operation is performed on the second qubit, if the first one is in the state $\ket 1$. If we interpret this gate from a 4-qudit perspective, we have a kind of SWAP gate that maps $\ket 2 \leftrightarrow \ket 3$. In essence we can regard the system of $n$ qubits as an $2^n$-qudit. However, it may be more illustrative to think of the system as consisting of $n$ qubits from the perspective of quantum computation.

\subsection{Generation of ZX Spider Operators}

Let us discuss the generation of $Z$-spiders and $X$-spiders of ZX-calculus. We consider $L_i$ input legs and $L_o$ output legs for qubits. For some parameter $\alpha$ the $Z$-spider is the following linear map: $\ket{0\cdots 0}\bra{0\cdots 0} + e^{i\alpha} \ket{1\cdots 1}\bra{1\cdots 1}$. In terms of multiway systems, this map corresponds to a non-interacting system of $2^{L_i}$ in-words and $2^{L_o}$ out-words where only two of in-words and two of out-words are connected.

As for $X$-spiders, they correspond to the following map for some parameter $\alpha$: $\ket{+\cdots +}\bra{+\cdots +} + e^{i\alpha} \ket{-\cdots -}\bra{-\cdots -}$. Here, $\ket \pm$ corresponds to eigenvectors of eigenvalue $\pm 1$ of the $\sigma_x$ operator. In particular $\ket{\pm} = (\ket 0 \pm \ket 1)/\sqrt 2$. We have seen that we can represent the Hadamard gate. If we apply Hadamard gate to each leg of a $Z$-spider, we can obtain the $X$-spider. In essence, an $X$-spider can be represented by using at least four layers in a multiway system.

\begin{figure}
    \centering
    \includegraphics[width=\linewidth]{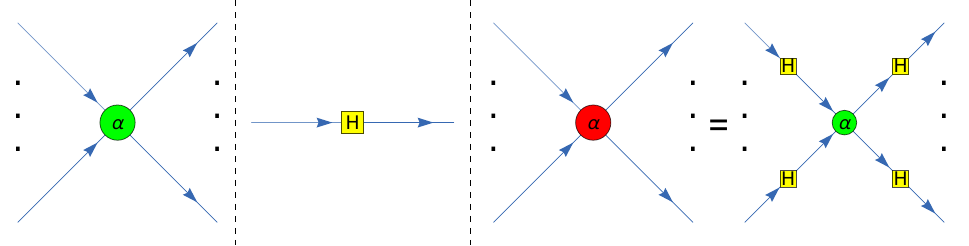}
    \caption{$ZX$-diagrams. (Left) $Z$-spider, (Middle) The Hadamard gate. (Right) $X$-spider.  }
    \label{fig:zx_diagrams}
\end{figure}

As a concrete example, let us consider the Hadamard gate, and $Z$-spider and $X$-spider for two input and two output legs. Their matrix representations are as follows:

\begin{eqnarray}
        H &=& \frac{1}{\sqrt 2}\begin{pmatrix}
        1 & 1\\
        1 & -1
    \end{pmatrix} 
\end{eqnarray}

\begin{eqnarray}    
    Z(\alpha) &=& \begin{pmatrix}
        1 & 0 & 0 & 0\\
        0 & 0 & 0 & 0\\
        0 & 0 & 0 & 0\\
        0 & 0 & 0 & e^{i\alpha}
    \end{pmatrix} 
\end{eqnarray}

\begin{eqnarray}  
    X(\alpha) &=& \frac 14 \begin{pmatrix}
        1+e^{i\alpha} & 1-e^{i\alpha} & 1+e^{i\alpha} & 1-e^{i\alpha}  \\
        1-e^{i\alpha} & 1+e^{i\alpha} & 1-e^{i\alpha} & 1+e^{i\alpha}\\
        1+e^{i\alpha} & 1-e^{i\alpha} & 1+e^{i\alpha} & 1-e^{i\alpha}\\
        1-e^{i\alpha} & 1+e^{i\alpha} & 1-e^{i\alpha} & 1+e^{i\alpha}
    \end{pmatrix}
\end{eqnarray}

For the corresponding minimal multiway systems see Figure~\ref{fig:zx_mw}.

\begin{figure}
    \centering
    \includegraphics[width=1.0\linewidth]{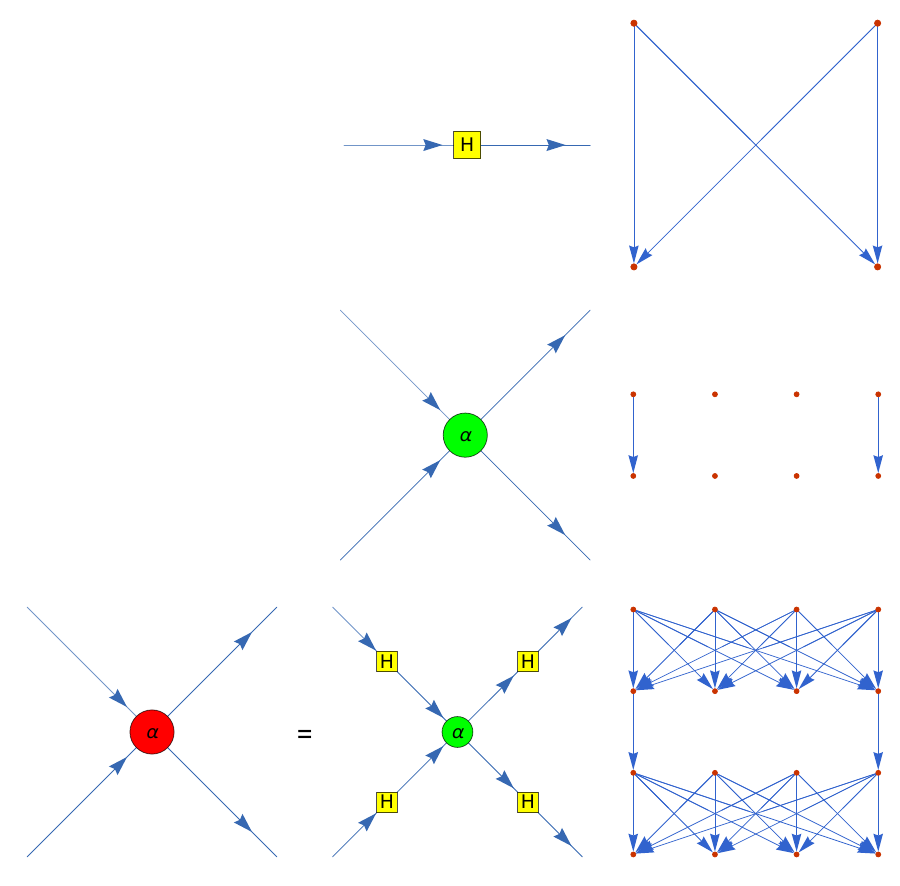}
    \caption{Minimal multiway system representations of the elements of $ZX$-calculus: the Hadamard gate, the $Z$-spider and the $X$-spider.}
    \label{fig:zx_mw}
\end{figure}

We can also obtain the CNOT or SWAP gates by permuting words in a non-interacting system. The discussion here indicates that one can find a (family of) ZX-diagrams starting from multiway systems.

\subsection{A Comment on Universal Quantum Computation}

So far we have shown that there is a correspondence between multiway systems and quantum circuits. The map from a multiway system to quantum circuits is one-to-many. This is because of the free parameters in the $S$-matrix one needs to determine. 

We have seen that we can represent the Hadamard gate and the $\pi/8$-gate for single qubits, as well as the CNOT gate for two qubits. What this means is that we can represent any given quantum circuit by suitably choosing a multiway system. On the other hand, since using a multiway system we can approximate any quantum circuit, we may say that the encoding of finite dimensional quantum systems by non-deterministic rewriting systems is universal. This means that using multiway systems one can represent a given quantum circuit. If the representation is not exact, one can consider a series of multiway systems, e.g. $\{\mathcal{M}_i | i\in I\}$ for some index set $I$ which may be $\mathbb N$, where in the limit (that is, $i\to\infty$) the representation approximates the given quantum circuit.

For a given multiway system, we can determine the words and their BSD variety values. Hence, using this information we can produce a \emph{dual} quantum circuit. In this case, we have the freedom to shuffle words at any layer as this is a symmetry of the multiway system and there is no a priori way to order words in each layer. This freedom in turn allows us to represent more quantum circuits for a given multiway system. As a simple example, the CNOT gate and the identity gate share the same connection structure as two layer multiway systems.

All in all, there exists a one-to-many map from multiway systems to quantum gates. For the case where the number of words in each layer is the same (\emph{e.g.} $2^n$ for some $n \in \mathbb Z^+$),  multiway systems yield representations of $H$, CNOT and $\pi/8$ gates. The compositions of these three gates realize universal quantum computation.

\section{Conclusions and Discussion } 
\label{sec:conclusion}

In this work, we have proposed a computational framework based on nondeterministic rewriting systems and shown that this encodes representations of finite-dimensional quantum operators, particularly those corresponding to quantum gates. The existence of such a map from abstract symbolic systems to quantum operators demonstrates the rich symbolic representational capabilities of abstract rewriting systems. This provides a background-independent framework for investigating quantum processes based on the nondeterministic causal structure of abstract computational systems. Rather than assuming an a priori Hilbert space of quantum operators, our work here takes a constructivist approach to building quantum operators. The  building blocks of our formalism are themselves not quantum states or operators, but instead are syntactic entities of an abstract computational system, without any quantum underpinnings. Nevertheless, these abstract symbolic systems yield representations of specific quantum gates and their compositions. 

More specifically, our formalism makes use of a rule-based string substitution system, involving a particular type of cyclic character strings called Leibnizian strings. We have shown that these strings exhibit a Fermi-Dirac distribution for expectation values of occupation numbers of character neighborhoods within the string. This is a combinatorial abstraction of an $N$-fermion system. Mathematically, this collection of character strings realizes  a $\mathbb{Z}$-module with a symmetric $\mathbb{Z}$-bilinear form. What we have constructed here is simply the discrete counterpart of a vector space with an inner product. This is the precise mathematical apparatus that admits the existence of a discrete analogue of a path integral based on an action defined over a sequence of Leibnizian strings, and the subsequent construction of a $S$-matrix for multiway rewriting systems.

In a sense, one may think of this as \emph{the statistical mechanics of computation}. By suitably defining a sum over paths in a discrete space of states, governed by a causal order, we have computed $S$-matrices that capture transition amplitudes between in-states and out-states. We have also shown that successive compositions of $S$-matrices, as multiway constructions, formally satisfy the rules of matrix multiplication. Furthermore, we find solutions to the path weights that ensure the unitarity of the $S$-matrices. It is these unitary $S$-matrices that yield representations of quantum operators. It is worth noting that the $S$-matrix generally identifies a family of quantum operators rather than a single specific one. We proceed to compute explicit examples of quantum gates involving qubits, qudits, and more general circuits. We find that, as  formal models of nondeterministic computation, rewriting systems of Leibnizian strings with causal structure encode representations of the CNOT, $\pi/8$, and Hadamard gates.

This brings us to an important question: how is it that an abstract rewriting system, which typically models classical computation, can have the capability to encode representations of quantum operators?    

First of all, let us note that although the multiway rewriting system is not itself quantum mechanical, its causal structure is nondeterministic. This is due to the partial order (poset) structure. What we have essentially done here is exploit this partial order to perform a statistical mechanics path-sum over sequences of rewriting events. This is a discrete analogue of the standard path integral over trajectories in phase space. The state transition amplitudes we compute are weighted over an ensemble of paths formed by a sequence of Leibnizian strings. Furthermore, unitarity is imposed by fixing the edge weights of the multiway system. Unitary ensures the necessary probability constraints over state transitions. This is how a complete representation of finite-dimensional quantum operators can be obtained from the structure of a poset with decorations (referring to node and edge data). 

Of course, a different choice of edge weights will not lead to unitary $S$-matrices. Solutions can be found that encode other types of matrices over the multiway path-sum. The multiway system is a  symbolic causal network, which akin to a neural network, can encode multiple representations based on its parameter choices. Our work here reveals that statistical mechanics on posets offers the possibility of a very rich encoding space. 

Interestingly, causal structure based on posets have prominently been used in approaches to quantum gravity - in particular, causal set theory \cite{surya2019causal, baron2025causal} and the causal fermions formalism \cite{finster2016continuum, fischer2025causal}. It would be remarkable to see whether our formalism here may complement the above-mentioned quantum gravity theories to operationalize statistical mechanics over posets - more specifically, to determine the local algebra of quantum observables on spacetime. It is worth noting that a different realization of statistical mechanics over causal structure has also recently been reported in \cite{kafker2025non}, where the investigation was focused on classical thermodynamic observables. Our methods here are complementary to that study.

To conclude, our investigations here lend further evidence towards the utility of examining quantum theory from symbolic or other inherently non-quantum underpinnings. Efforts involving systems that are deterministic \cite{thooft}, classical statistical \cite{bondar2012operational, klein2018koopman, klein2020probabilistic, pehle2018emulating, wetterich2024probabilistic}, and pre-geometric \cite{fiorenza2013higher, arsiwalla2024operator, chester2024quantization, chester2025preons, tavanfar2018holoquantum} have already been fruitful in extracting key features of quantum theory. This also suggests avenues for a broader realization of quantum process theories and quantum-like phenomena (for instance, see \cite{arsiwalla2024operator, tavanfar2023unitary, tavanfar2025does}). Interestingly, non-spatiotemporal formalisms of quantum process theories that are based on monoidal categorical structures and diagrammatic reasoning have also found applications to a wide range of domains from natural language semantics \cite{coecke2010mathematical, coecke2021mathematics} to quantum cognition \cite{signorelli2018moral}.  But beyond applications, non-spatiotemporal formalisms of quantum theory strongly point towards the deeper question of what the bare essence of ``quantum-ness'' should be, both, for physical as well as abstract systems like language or cognition; and whether the conventional algebraic structures of quantum theory themselves have a pre-quantum origin (see \cite{arsiwalla2024operator} for one such formulation). A plausible perspective on the latter is that \emph{finite-dimensional quantum operators may be realized via the statistical mechanics of computation}.



\section*{Conflicts of Interest}

The authors declare no conflicts of interest.

\section*{Data Availability Statement}

Code related to the present study can be accessed on GitHub:

\begin{quote}
    \url{https://github.com/fsdundar/qstr-monads}
\end{quote}

\bibliographystyle{eptcs}
\bibliography{refs}

\end{document}